\documentclass{article}

\usepackage{arxiv}
\usepackage[utf8]{inputenc}
\usepackage[T1]{fontenc}
\usepackage{amsmath}
\usepackage{amssymb}
\usepackage{amsthm}
\usepackage[hypertexnames=false]{hyperref}
\usepackage{url}
\usepackage{booktabs}
\usepackage{microtype}
\usepackage{graphicx}
\usepackage{natbib}
\usepackage{doi}
\usepackage{algorithm}
\usepackage{algpseudocode}
\usepackage{tikz}
\usepackage{dsfont}

\usetikzlibrary{positioning, arrows.meta, decorations.pathreplacing}

\newcommand{\ind}{\mathord{\mathds{1}}}
\newcommand{\orcidicon}[1]{\href{https://orcid.org/#1}{\includegraphics[scale=0.06]{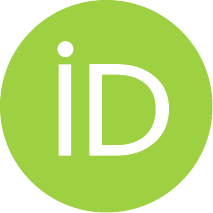}\hspace{1mm}}}

\newtheorem{theorem}{Theorem}
\newtheorem{proposition}{Proposition}
\newtheorem{lemma}{Lemma}
\newtheorem{corollary}{Corollary}
\newtheorem{definition}{Definition}
\newtheorem{example}{Example}
\newtheorem{remark}{Remark}

\title{Closed-Form Predicate-Level Shapley Attribution for Sliding-Window Aggregates}
\date{}
\author{%
  \orcidicon{0000-0002-4786-0572}Pouya Khani \\
  Department of Computer Science\\
  Aarhus University\\
  Denmark \\
  \texttt{pouya.khani@cs.au.dk} \\
  \And
  \orcidicon{0000-0002-1091-9948}Ira Assent \\
  Department of Computer Science\\
  Aarhus University\\
  Denmark \\
  \texttt{ira@cs.au.dk} \\
}

\renewcommand{\shorttitle}{Closed-Form Predicate-Level Shapley Attribution}

\hypersetup{
  pdftitle={Closed-Form Predicate-Level Shapley Attribution for Sliding-Window Aggregates},
  pdfsubject={Computer Science - Databases},
  pdfauthor={Pouya Khani and Ira Assent},
  pdfkeywords={Explainable AI, Data Streaming, Shapley Values, Predicate Attribution, Sliding Windows}
}

\begin{document}
\setlength{\emergencystretch}{2em}
\maketitle

\begin{abstract}
Streaming engines report sliding-window aggregates in real time, but they do not explain \emph{why} an aggregate takes its current value. A natural target is the Shapley value from cooperative game theory, which axiomatically distributes an aggregate among the tuples in the window. Practitioners, however, ask predicate-level questions (e.g., how much a region or customer tier contributed to an average or variance spike). Exact Shapley computation is exponential in the window size, and existing estimators discard the massive overlap between consecutive windows.

We show that for SUM, COUNT, AVG, and population/sample variance, exact predicate-level Shapley values admit closed forms in three additively maintained summaries per predicate (count, sum, and sum of squares), with coefficients that depend only on two running harmonic numbers. Attribution therefore reduces to $O(1)$ summary updates per slide for registered predicates, with no coalition enumeration. Overlapping and compositional predicates are answered exactly via atomic refinement of Boolean signatures. We further characterize the phenomenon: every moment-polynomial aggregate admits such a form, while MAX, MIN, and quantiles provably do not at any fixed moment order.

Experiments match brute-force Shapley values to floating-point precision on over $10{,}000$ windows, sustain $\approx\!2\,\mu$s per slide up to $N=10^6$ ($3{,}200\times$ faster than per-window recomputation of the same formulas), and explain a nighttime fare spike on 2.9M NYC taxi trips at $\approx\!1.8$M summary updates per second.
\end{abstract}


\section{Introduction}

High-velocity data streams, from algorithmic trading and IoT monitoring to real-time cybersecurity, have made Data Stream Management Systems (DSMS) a standard tool. Modern engines such as Apache Flink and Spark Streaming compute continuous aggregate queries over sliding windows efficiently. These systems can instantly report \textit{what} is happening (e.g., ``the average server latency has spiked by 200\,ms''), but they cannot explain \textit{why} it is happening in real time.

Standard root-cause analysis uses ad-hoc manual queries or post-hoc analysis on data warehouses. These approaches introduce unacceptable latency. To automate this, we turn to the Shapley Value (SV), a concept from cooperative game theory. Unlike heuristic feature importance methods, the Shapley Value distributes the ``payout'' (the aggregate result) among the ``players'' (the input tuples) axiomatically. It satisfies properties such as Efficiency, Symmetry, Additivity, and the Dummy axiom.

\begin{example}[Running Example: Cloud Latency Monitoring]
Consider an operations team monitoring request latency for a cloud service. A streaming engine continuously reports AVG and VAR over the most recent requests. When the average latency spikes, the team asks whether the increase is mainly due to requests from \texttt{region='EU-West'}, premium customers, a particular service endpoint, or an intersection such as \texttt{region='EU-West' AND tier='premium'}. The relevant explanation must be produced while the window is still active. It should also distinguish a predicate that increases the average from one that merely contributes many near-average tuples. We use this latency-monitoring scenario throughout the paper. Small numerical examples use normalized latency values such as $W=\{2,4,10\}$ to keep the arithmetic readable.
\end{example}

\subsection{The Challenges}
Despite its theoretical superiority, applying Shapley Values to streaming aggregates presents three prohibitive challenges:

\textbf{1. The Re-computation Bottleneck:} For a generic cooperative game, exact Shapley computation requires summing over exponentially many coalitions. Sampling-based approximations still require many marginal evaluations. In a sliding-window context, where $W_t$ denotes the active window at time $t$ and $W_{t+1}$ the next window, the two window states largely overlap. Stateless methods such as Monte Carlo permutation sampling discard useful work from the previous window and repeat many nearly identical computations.

\textbf{2. The Granularity Mismatch (Tuple vs. Predicate):} Existing literature mainly focuses on the Shapley value of individual atomic tuples (row-level attribution). However, in streaming analytics, users rarely ask about a specific raw event. Instead, they require \textit{Predicate-Based Attribution}. A user is more likely to ask, ``How much did \texttt{Region='EU-West'} contribute to the variance?'' than to query a specific transaction ID. A predicate $\mathcal{P}$ is a Boolean condition over tuple attributes. Computing attribution for $\mathcal{P}$ treats the dynamic subset of tuples that satisfy $\mathcal{P}$ as a coalition. In a streaming context, the cardinality of this subset changes at every slide, complicating standard group-testing approaches. Users also frequently issue \emph{compositional} questions that involve conjunctions, disjunctions, and negations of multiple predicates. An attribution scheme that only supports single-predicate queries is therefore insufficient.

\textbf{3. Non-Linearity in Aggregates:} Linear aggregates (SUM, COUNT) are trivially decomposable. Non-linear functions such as AVG and VAR introduce interaction effects: the marginal contribution of one tuple depends on the other tuples present in the coalition. The challenge is to show that these interactions can still be summarized exactly for important SQL aggregates whose Shapley game admits a closed form in low-order moments, rather than approximating the full Shapley game.

\subsection{Contributions}
We derive exact closed-form predicate-level Shapley formulas for sliding-window aggregates whose cooperative game depends only on low-order moments: SUM, COUNT, AVG, population variance, and sample variance. The formulas need only the window's size and its first and second moments, all of which are easy to keep up to date. Once the summaries are updated for tuples entering and leaving the window, attribution is a single formula evaluation. No tuple-level Shapley computation and no coalition enumeration are needed.

\textbf{Notation preview.} We use $W_t$ for the active window at time $t$ and $W$ for a generic active window when the time index is not important. We use $\mathcal{P}$ for a user predicate and $C_{\mathcal{P}}$ for the active tuples satisfying $\mathcal{P}$. The predicate summaries are $m_{\mathcal{P}}=|C_{\mathcal{P}}|$, $A_{\mathcal{P}}=\sum_{x_i\in C_{\mathcal{P}}}y_i$, and $B_{\mathcal{P}}=\sum_{x_i\in C_{\mathcal{P}}}y_i^2$, where $y_i$ is the numeric attribute being aggregated. We write $\Delta=(x_{out},x_{in})$ for a single-tuple count-window slide and $\Delta_t=(X_t^{out},X_t^{in})$ for a time-window slide with outgoing and incoming tuple sets. Full notation is summarized in Table~\ref{tab:notation}.

\begin{itemize}
    \item \textbf{Exact Sufficient-Statistics Shapley Formulas:} We derive exact closed-form Shapley expressions for SUM, COUNT, AVG, population variance, and sample variance over a finite window. We show that predicate-level attribution for each is a linear function of compact predicate summaries. For a single-relation static database, the tuple-level AVG formula matches an independent result of Standke and Kimelfeld~\cite{standke2025tractability} (Section~\ref{sec:avg}). The variance formulas, the predicate-level forms, and their windowed maintenance are new. Every attribution coefficient is a simple expression in two harmonic numbers and can be updated in $\mathcal{O}(1)$ as the window size changes (Section~\ref{sec:variance}).
    \item \textbf{Predicate Attribution without Tuple Materialization:} We show that predicate questions can be answered exactly from $(m_{\mathcal{P}}, A_{\mathcal{P}}, B_{\mathcal{P}})$ together with global summaries, without computing or storing any tuple-level Shapley value.
    \item \textbf{Sliding-Window Summary Maintenance:} We show how each predicate summary changes under a slide, both for single-tuple count-based deltas $\Delta=(x_{out},x_{in})$ and for multi-tuple time-based deltas $\Delta_t=(X^{out}_t,X^{in}_t)$, avoiding re-evaluation of the full cooperative game on consecutive windows.
    \item \textbf{Time-Based Window Generalization:} We extend the closed forms to time-based windows where the active cardinality $N_t$ varies, with explicit batched delta-update rules; the harmonic-number coefficients follow $N_t$ at $\mathcal{O}(1)$ cost.
    \item \textbf{Overlapping and Compositional Predicates:} We support overlapping predicates and arbitrary Boolean combinations of them (intersections, unions, complements, and set differences). We do so via the \emph{atomic refinement} of their Boolean signatures. Maintaining sufficient statistics per atom yields exact attribution for any compositional query and restores Shapley efficiency over the resulting window partition.
    \item \textbf{Additivity Across Disjoint Subsets:} The required sufficient statistics are additive across disjoint subsets, so the closed forms compose cleanly with existing pane-based sliding-window aggregation frameworks~\cite{li2005pane,tangwongsan2015general,traub2021scotty}.
    \item \textbf{Boundary of the Phenomenon:} We prove that the approach extends to every \emph{moment-polynomial} window game, i.e., any game that is a polynomial in the coalition's power sums with size-dependent coefficients (Section~\ref{sec:boundary}). We also prove the converse boundary: MAX, MIN, and fixed order statistics have \emph{no} sufficient-statistic representation of any fixed moment order, because windows can match all moments up to order $d$ yet differ in their extreme values.
    \item \textbf{Semantics Made Explicit:} We show that AVG attribution splits exactly into a \emph{group share} plus a harmonically weighted \emph{group lift} (Section~\ref{sec:scale}). This reveals that attributions grow logarithmically with window size. On real data, we measure how often the Shapley weighting changes the explanation compared to the naive lift heuristic, and how far the sum-of-members semantics is from the predicate-as-player alternative (Sections~\ref{sec:eval:lift} and~\ref{sec:eval:casestudy}).
    \item \textbf{Experimental Evaluation:} We validate exactness against brute-force enumeration on $10{,}010$ random windows. We measure constant per-slide latency up to $N=10^6$, comparing against per-window recomputation, permutation sampling~\cite{castro2009polynomial}, and exact enumeration. A case study on 2.9 million NYC taxi trips~\cite{nyctlc2024} explains a nighttime fare spike as long-haul airport traffic while sustaining $\approx\!1.8$ million summary updates per second in a single-threaded prototype (Section~\ref{sec:eval}).
\end{itemize}

\section{Related Work}

The problem of explaining aggregate query results in data streams intersects three research communities: Cooperative Game Theory, Machine Learning (Data Valuation), and Database Provenance. We review the state of the art in each and analyze their limitations regarding the sliding-window predicate-attribution problem.

\subsection{Stateless Shapley Approximations}
The foundational work by Shapley \cite{shapley1953value} established the axioms of Efficiency, Symmetry, and Additivity. For a generic game represented by a value oracle, exact computation requires summing over exponentially many coalitions. To address this, Castro et al.\ \cite{castro2009polynomial} introduced sampling-based estimation using the theory of permutations. This approach was unified and popularized in the Machine Learning community by \textbf{SHAP} (Shapley Additive Explanations) \cite{lundberg2017unified}. KernelSHAP, in particular, approximates values by solving a weighted linear regression on perturbed inputs.

\textbf{Limitation:} These methods are designed for static datasets and do not exploit aggregate-specific structure. In a streaming context, a sliding-window transition $W_t \to W_{t+1}$ involves a high degree of data overlap (typically $>99\%$). Stateless algorithms such as KernelSHAP treat $W_{t+1}$ as a completely new dataset, discarding the computational effort spent on $W_t$. They also operate at the tuple level, so even if applied to a stream they would not directly yield predicate-level answers.

\subsection{Variance-Based Sensitivity and Shapley Effects}
Variance attribution has also been studied in global sensitivity analysis through Sobol' indices and Shapley effects \cite{owen2014sobol,song2016shapley}. That line of work allocates the variance of a model output among input \emph{variables}, typically under a probabilistic input model. The structural connection is that both formulations exploit the fact that variance functionals depend on the data only through low-order moments. Our finite-window expressions can be viewed as the discrete, sampling-without-replacement counterpart of the continuous Shapley-effects framework.

\textbf{Limitation:} The players, the object of explanation, and the query type differ from ours. Sobol'/Shapley-effects allocate variance among input \emph{variables} of a probabilistic model. In contrast, we attribute a database aggregate over a realized finite window to \emph{tuples} and \emph{predicates}. Those methods therefore do not yield sliding-window predicate attributions for SQL-style aggregates.

\subsection{Incremental Shapley and XAI on Streams}
Recent work has adapted Shapley-style explanations to streaming or temporal settings. \textbf{DeltaSHAP} \cite{kim2025deltashap} explains \textit{prediction evolution} in online ML (e.g., patient monitoring). It attributes \textit{changes in model predictions} over time using Shapley-style reasoning. \textbf{iSAGE} \cite{muschalik2023isage} and \textbf{iPFI} \cite{fumagalli2023ipfi} provide incremental \textit{feature} importance on data streams and handle concept drift.

\textbf{Limitation:} These methods explain model outputs (e.g., risk scores) or which \textit{features} drive predictions. They do not explain the result of a streaming \textit{aggregate query} (SUM, AVG, VAR over windows) attributed to \textit{predicates} (e.g., region, sector). Their maintenance is driven by changes in model predictions, whereas our closed forms are driven by the window delta.

\subsection{Incremental Shapley and Data Valuation}
Recent work in ``Data Valuation'' seeks to quantify the value of training data points for ML models. Ghorbani and Zou \cite{ghorbani2019data} introduced the Data Shapley to identify low-quality data. To avoid expensive model retraining, Jia et al.\ \cite{jia2019efficient} proposed efficient incremental algorithms for K-Nearest Neighbors (\textit{KNN-Shapley}). They later extended this to heuristic approximations for deep learning updates.

\textbf{Limitation:} These methods largely assume an \textit{append-only} or \textit{batch-update} setting suitable for model retraining pipelines. They handle the insertion of a new tuple ($+x_{new}$). They are not designed around the \textit{eviction} phase of a sliding window ($-x_{old}$). Some incremental valuation methods also maintain model- or history-dependent state \cite{yan2021monitoring}. Our closed forms instead target aggregate-query windows, where simultaneous insertion and eviction are absorbed by additive updates to bounded predicate summaries.

\subsection{Explanations in Data Management}
The database community has explored query explanations through provenance, causality, responsibility, group-by-style explanations, and sensitivity analysis. \textit{Scorpion} \cite{wu2013scorpion} uses backward tracing and influence-style reasoning to identify tuples responsible for outliers in aggregate queries. MacroBase \cite{bailis2017macrobase} prioritizes attention in fast data by highlighting important and unusual behavior in streams. DIFF \cite{abuzaid2021diff} provides a relational interface for large-scale data explanation. \textit{TSExplain} \cite{chen2023tsexplain} segments aggregated time series so that each segment admits a stable set of top contributors. It targets evolving contributor explanations for KPI curves rather than axiomatic Shapley attribution of a window aggregate to logical predicates. \textit{Erebus} \cite{palyvos2022erebus} explains divergences between expected and observed streaming query answers (including missing outputs) using provenance-style reasoning under storage limits. It complements forward why-provenance. Glavic et al.\ \cite{glavic2013ariadne} instrument stream operators to propagate fine-grained tuple lineage in DSMSs, trading off eager versus lazy provenance for overhead. Roy and Suciu \cite{roy2014formal} formalize explanations for database queries through intervention-style changes. Meliou et al.\ \cite{meliou2010causality} formalized causality and responsibility for query answers and non-answers. More recent work studies Shapley values directly in database query answering. Livshits et al.\ \cite{livshits2020shapley} initiated the complexity study of Shapley contributions of tuples to query answers. They established a dichotomy for conjunctive queries and first results for aggregate queries over static databases. Subsequent work develops practical computation via knowledge compilation and model counting \cite{deutch2022computing,kara2024shapley,bertossi2023shapleydb}, tractable alternative responsibility measures \cite{bienvenu2025shapley}, lineage-based algorithms for aggregate queries \cite{abramovich2025advancing}, and a tractability map for aggregate conjunctive queries \cite{standke2025tractability}. Closest to our closed forms, Standke and Kimelfeld \cite{standke2025tractability} derive the tuple-level Shapley value of AVG for a single-relation query over a static database (their Proposition~5.2). This coincides with our Eq.~\ref{eq:avg_tuple}; see Section~\ref{sec:avg}. That line of work is tuple-level and static. It addresses neither predicate-level attribution, nor variance games, nor maintenance under sliding windows.

\textbf{Limitation:} These works provide important foundations for attribution in databases. They do not directly address the sliding-window predicate-attribution problem studied here. In prior work we studied predicate-level attribution on static tables under partial causal knowledge via Causal Banzhaf Value (CBV)~\cite{khani2025causal} and Causally-Constrained Power Indices (CCPI)~\cite{khani2026ccpi}. Those methods constrain static power-index explanations; they do not provide closed-form streaming attribution. Causality and sensitivity methods are not generally Shapley-axiomatic. Shapley-based query-answer explanations are formulated for static database instances rather than continuously changing windows. Their complexity results concern worst-case query structure rather than the moment structure of aggregate games that we exploit for constant-time maintenance. Our work targets predicate-level Shapley attribution for streaming aggregate queries. It exploits overlap between consecutive windows in systems such as Apache Flink \cite{carbone2015apache}.

\subsection{Sliding-Window Aggregates, Intermittency, and Shared Maintenance}
General-purpose window aggregation frameworks such as panes \cite{li2005pane}, the Reactive Aggregator \cite{tangwongsan2015general}, and Scotty \cite{traub2021scotty} maintain associative aggregates efficiently under arbitrary window specifications. Our sufficient statistics are associative and slot directly into these frameworks (Corollary~\ref{cor:panes}). Shein and Chrysanthis \cite{shein2022multi} optimize multi-query plans for incrementally evaluated sliding-window aggregates, sharing work across concurrent window specifications. Zhu and Ravishankar \cite{zhu2004intermittent} approximate aggregate answers when streams are \emph{intermittent}, using randomized summaries under burstiness assumptions.

\textbf{Limitation:} These systems target scalable \emph{evaluation} of window aggregates, not axiomatic \emph{explanation} of the result. They do not produce Shapley attributions to predicates. Intermittent-stream methods further assume incomplete observations and approximate answers. Our exact formulas require a fully observed active window. Sketch-based or intermittent settings are left to the approximate extensions discussed in our limitations.

\section{Background and Preliminaries}

In this section, we formalize the data stream model, the semantics of sliding-window aggregation, and the game-theoretic foundations of the Shapley Value. We then extend Shapley attribution from atomic players to logical predicates.

\subsection{Data Stream and Sliding Windows}
Let $\mathcal{S} = \{x_1, x_2, \dots \}$ be an unbounded stream of tuples. Each tuple $x_i \in \mathcal{D}$ consists of a set of features and a timestamp $\tau_i$. We consider two window models.

\paragraph{Count-based windows.} A count-based window of size $N$ contains the $N$ most recent tuples:
\begin{equation}
    W_t = \{x_{t-N+1}, \dots, x_t\}, \qquad |W_t|=N.
\end{equation}
The transition from $W_t$ to $W_{t+1}$ is defined by a delta pair $\Delta = (x_{out}, x_{in})$. Here $x_{out} = x_{t-N+1}$ is evicted and $x_{in} = x_{t+1}$ is inserted. Cardinality is fixed.

\paragraph{Time-based windows.} A time-based window of duration $\tau$ contains all tuples whose timestamp lies in the trailing interval $(t-\tau,t]$:
\begin{equation}
    W_t = \{x_i \in \mathcal{S} \mid \tau_i \in (t-\tau,t]\}, \qquad |W_t|=N_t.
\end{equation}
The active cardinality $N_t$ varies with the arrival process. The transition from $W_t$ to $W_{t'}$ at the next emit instant $t'$ is described by a multi-tuple delta
\begin{equation}
    \Delta_{t,t'} = (X^{out}_{t,t'},\,X^{in}_{t,t'}),
\end{equation}
where $X^{in}_{t,t'} = \{x_i : \tau_i \in (t,t']\}$ and $X^{out}_{t,t'} = \{x_i : \tau_i \in (t-\tau,t'-\tau]\}$. We write $|\Delta_{t,t'}| = |X^{in}_{t,t'}|+|X^{out}_{t,t'}|$. We treat both window types uniformly when possible. We write $W$ for an unspecified window of cardinality $N$ (or $N_t$ for time-based contexts where the dependence matters). Figure~\ref{fig:window_delta} illustrates a single-tuple count-based slide. The time-based case generalizes both arrows to sets.

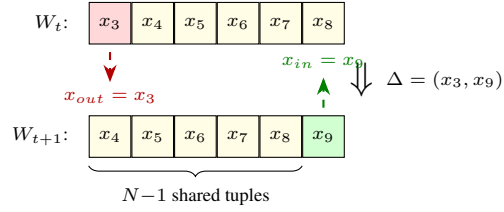
\begin{figure}
\centering
\begin{tikzpicture}[
    >=Stealth,
    cell/.style={draw, minimum width=0.55cm, minimum height=0.55cm, font=\scriptsize, inner sep=1pt},
    lbl/.style={font=\scriptsize},
    arr/.style={->, thick}
]
\node[lbl, anchor=east] at (-0.15, 1.5) {$W_t$:};
\node[cell, fill=red!15] (a1) at (0.275, 1.5) {$x_3$};
\node[cell, fill=yellow!12, right=0pt of a1] (a2) {$x_4$};
\node[cell, fill=yellow!12, right=0pt of a2] (a3) {$x_5$};
\node[cell, fill=yellow!12, right=0pt of a3] (a4) {$x_6$};
\node[cell, fill=yellow!12, right=0pt of a4] (a5) {$x_7$};
\node[cell, fill=yellow!12, right=0pt of a5] (a6) {$x_8$};

\draw[->, red!70!black, thick, dashed] ([yshift=-0.1cm]a1.south) -- ++(0, -0.4)
    node[below, font=\scriptsize, text=red!70!black] {$x_{out}=x_3$};

\node[font=\Large] at (3.6, 0.75) {$\Downarrow$};
\node[lbl] at (4.7, 0.75) {$\Delta=(x_3, x_9)$};

\node[lbl, anchor=east] at (-0.15, 0) {$W_{t+1}$:};
\node[cell, fill=yellow!12] (b1) at (0.275, 0) {$x_4$};
\node[cell, fill=yellow!12, right=0pt of b1] (b2) {$x_5$};
\node[cell, fill=yellow!12, right=0pt of b2] (b3) {$x_6$};
\node[cell, fill=yellow!12, right=0pt of b3] (b4) {$x_7$};
\node[cell, fill=yellow!12, right=0pt of b4] (b5) {$x_8$};
\node[cell, fill=green!18, right=0pt of b5] (b6) {$x_9$};

\draw[->, green!50!black, thick, dashed] ([yshift=0.1cm]b6.north) -- ++(0, 0.4)
    node[above, font=\scriptsize, text=green!50!black] {$x_{in}=x_9$};

\draw[decorate, decoration={brace, amplitude=4pt, mirror}]
    ([yshift=-3pt]b1.south west) -- ([yshift=-3pt]b5.south east)
    node[midway, below=5pt, font=\scriptsize] {$N{-}1$ shared tuples};
\end{tikzpicture}
\caption{Single-tuple count-based slide. The window $W_t$ advances by evicting $x_{out}$ and inserting $x_{in}$. Of $N$ tuples, $N{-}1$ are shared between consecutive windows. The closed forms below therefore need only update the summaries affected by the two boundary tuples. For time-based windows, $x_{out}$ and $x_{in}$ generalize to sets $X^{out}_t$ and $X^{in}_t$ of arbitrary size.}
\label{fig:window_delta}
\end{figure}

\subsection{The Shapley Value}
The Shapley Value is a solution concept from cooperative game theory. It assigns a unique distribution of the total surplus generated by a coalition of players. In the context of data attribution, the players are the tuples in the window $W$, and the payout is the aggregate result. Formally, let $\nu:2^{W}\to\mathbb{R}$ denote the cooperative game induced by the aggregate. For any coalition $S\subseteq W$, $\nu(S)$ is the value of the aggregate computed on the tuples in $S$ (with the convention $\nu(\emptyset)=0$ unless stated otherwise). Concrete instances include $\nu_{\mathrm{SUM}}$, $\nu_{\mathrm{AVG}}$, and $\nu_{\mathrm{VAR}}$ defined in Section~\ref{sec:closedforms}.

For a tuple $x_i \in W$, the Shapley value $\phi_i(\nu, W)$ is the weighted average of its marginal contributions to all possible subsets $S \subseteq W \setminus \{x_i\}$:
\begin{equation}
    \phi_i(\nu, W) = \sum_{S \subseteq W \setminus \{x_i\}} \frac{|S|! (|W| - |S| - 1)!}{|W|!} [\nu(S \cup \{x_i\}) - \nu(S)].
    \label{eq:shapley_def}
\end{equation}
This formulation satisfies the standard axioms:
\begin{itemize}
    \item \textbf{Efficiency:} $\sum_{i \in W} \phi_i = \nu(W) - \nu(\emptyset)$.
    \item \textbf{Symmetry:} If two tuples $x_i, x_j$ contribute equally to all coalitions, then $\phi_i = \phi_j$.
    \item \textbf{Additivity:} For two games $\nu_1, \nu_2$, $\phi_i(\nu_1 + \nu_2) = \phi_i(\nu_1) + \phi_i(\nu_2)$.
    \item \textbf{Dummy:} If a tuple has zero marginal contribution to every coalition, then $\phi_i = 0$.
\end{itemize}

\subsection{Predicate-Based Attribution}
The classical Shapley literature focuses on the contribution of atomic tuples. Streaming analytics often requires explanations at the level of logical predicates (e.g., ``\texttt{sector = technology}''). Let $\mathcal{P}: \mathcal{D} \to \{0, 1\}$ be a boolean predicate. We define the \textit{Predicate Coalition} $C_{\mathcal{P}} \subseteq W$ as the set of all active tuples satisfying the predicate:
\begin{equation}
    C_{\mathcal{P}} = \{x \in W \mid \mathcal{P}(x) = 1\}.
\end{equation}
The attribution for the predicate, denoted $\Phi_{\mathcal{P}}$, is the sum of the atomic Shapley values of its constituents under the same aggregate game $\nu$:
\begin{equation}
    \Phi_{\mathcal{P}}(W) = \sum_{x_i \in C_{\mathcal{P}}} \phi_i(\nu, W).
    \label{eq:predicate_def}
\end{equation}

\textbf{Partition, overlap, and atomic refinement.}
By the Efficiency axiom, if a family $\{\mathcal{P}_1,\dots,\mathcal{P}_K\}$ partitions the window, then $\sum_{k} \Phi_{\mathcal{P}_k}(W) = \nu(W) - \nu(\emptyset)$. That is, the predicate attributions form an exhaustive, mutually exclusive decomposition of the aggregate. For a single predicate $\mathcal{P}$ together with its complement $\neg\mathcal{P}$, this reduces to $\Phi_{\mathcal{P}} + \Phi_{\neg\mathcal{P}} = \nu(W)-\nu(\emptyset)$. For \emph{overlapping} predicates, the per-predicate quantity $\Phi_{\mathcal{P}}$ remains the exact aggregate Shapley contribution of $C_{\mathcal{P}}$. The family $\{\Phi_{\mathcal{P}_k}\}$ no longer sums to $\nu(W)-\nu(\emptyset)$, because tuples in multiple predicates are double-counted. We address overlap in Section~\ref{sec:overlap} by working at the level of the \emph{atomic refinement} of the predicates.

\textbf{Grand-coalition choice.} Throughout the paper, the grand coalition is the active window $W_t$. Attribution therefore explains the current aggregate relative to the tuples currently retained by the window. Other choices (e.g., all tuples seen so far, or a larger historical horizon) define different cooperative games and yield different Shapley values.

\textbf{The Challenge.} In a sliding window, the composition of $C_{\mathcal{P}}$ is dynamic. Both the player set $W$ and the subset $C_{\mathcal{P}}$ change at every slide. Naively recomputing Eq.~\ref{eq:predicate_def} requires evaluating Eq.~\ref{eq:shapley_def} for every member of the changing group. This is computationally prohibitive. The goal of the present work is to identify aggregate classes for which $\Phi_{\mathcal{P}}$ can be expressed in closed form from small, additively-maintained summaries.

\subsection{Aggregate Classification}
The complexity of computing $\phi_i$ depends on the properties of $\nu$. We classify streaming aggregates into two categories:
\begin{itemize}
    \item \textbf{Linear/Decomposable (e.g., SUM, COUNT):} The marginal contribution is independent of the coalition $S$.
    \item \textbf{Sufficient-Statistic Non-Linear (e.g., AVG, VAR):} The marginal contribution depends on the tuples in $S$. However, the expectation over uniformly random predecessor coalitions can be expressed using low-order sufficient statistics of $W$.
\end{itemize}
Our results address both decomposable aggregates and non-linear aggregates whose Shapley expressions reduce to incrementally maintained sufficient statistics. They do \emph{not} cover aggregates such as quantiles, MIN/MAX, or arbitrary user-defined functions. Section~\ref{sec:boundary} makes both sides of this boundary precise.

\begin{table}[t]
\scriptsize
\setlength{\tabcolsep}{2pt}
\renewcommand{\arraystretch}{1.15}
\centering
\begin{tabular}{|p{0.30\linewidth}|p{0.58\linewidth}|}
\hline
\textbf{Symbol} & \textbf{Meaning} \\ \hline
$W_t,W,N,N_t$ & Active window, generic window, fixed count-window size, and active time-window size \\ \hline
$x_i,y_i$ & Tuple $i$ and its numeric value for SUM/AVG/VAR \\ \hline
$\Delta,\Delta_t$ & Single-tuple count-window slide and multi-tuple time-window slide \\ \hline
$\nu,\phi_i$ & Aggregate game and tuple-level Shapley value \\ \hline
$\mathcal{P},C_{\mathcal{P}},\Phi_{\mathcal{P}}$ & Predicate, its active coalition, and predicate attribution \\ \hline
$A(W),B(W)$ & Global sum and squared sum over $W$ \\ \hline
$m_{\mathcal{P}},A_{\mathcal{P}},B_{\mathcal{P}}$ & Predicate count, sum, and squared sum \\ \hline
$H_N,H_N^{(2)}$ & First- and second-order harmonic numbers \\ \hline
$T_i,U_i,E_i^{(2)}(r)$ & Auxiliary per-tuple variance quantities defined in Section~\ref{sec:variance} \\ \hline
$U_{\mathcal{P}},YT_{\mathcal{P}},Q_{\mathcal{P}}$ & Predicate-level variance combinations defined in Eq.~\ref{eq:UP}--\ref{eq:QP} \\ \hline
$C_U,C_Q,C_{YT},C_B$ & Population-variance coefficients depending only on $N$ \\ \hline
$C_U^S,C_Q^S,C_{YT}^S,C_B^S$ & Sample-variance coefficients depending only on $N$ \\ \hline
$\mathrm{sig}(x),\alpha_s$ & Boolean signature of tuple $x$ and corresponding atom \\ \hline
\end{tabular}
\caption{Frequently used notation. Symbols used only locally are defined where they appear.}
\label{tab:notation}
\end{table}

\section{Closed-Form Predicate-Level Shapley Attribution}
\label{sec:closedforms}

We now derive exact closed-form predicate-level Shapley attribution for the aggregate classes above. Let the transition from $W_t$ to $W_{t+1}$ be defined by the delta $\Delta_t$. For count-based windows, $\Delta_t$ is the singleton $(x_{out},x_{in})$. For time-based windows, it is the set pair $(X^{out}_t,X^{in}_t)$. Throughout this section we focus on a single predicate $\mathcal{P}$. The multi-predicate and compositional case is treated in Section~\ref{sec:overlap}.

\subsection{Attribution Semantics}
\label{sec:semantics}
The word ``delta'' in our framing refers to the \emph{maintenance} delta of the sliding window. This delta consists of tuples leaving the active state while others enter. The attribution \emph{semantics}, by contrast, concern the current-window quantity $\Phi_{\mathcal{P}}(W_t)$. We do not introduce a separate Shapley game over the change $\nu(W_{t+1})-\nu(W_t)$. Defining such a ``delta game'' would require an additional modeling choice, because $W_t$ and $W_{t+1}$ have different player sets. A semantic theory of change attribution is left open (Section~\ref{sec:limitations}).

\subsection{Decomposable Aggregates: SUM and COUNT}
For SUM over a numeric attribute $y$, the cooperative game is $\nu(S)=\sum_{x_i \in S} y_i$. The marginal contribution of tuple $x_i$ is always $y_i$, independent of the predecessor coalition. Hence $\phi_i=y_i$. Predicate attribution is simply
\begin{equation}
    \Phi_{\mathcal{P}}^{SUM}(W) = \sum_{x_i \in W} \ind_{\mathcal{P}}(x_i)\,y_i = A_{\mathcal{P}}.
\end{equation}
COUNT is analogous. Each tuple contributes one unit to the count game, so the attribution of a predicate is the number of active tuples satisfying it: $\Phi_{\mathcal{P}}^{COUNT}(W) = m_{\mathcal{P}}$. Both SUM and COUNT are maintained exactly by incrementing the predicate summary for every tuple in $X^{in}_t$ and decrementing it for every tuple in $X^{out}_t$.

\subsection{Average}
\label{sec:avg}
AVG is non-additive but admits an exact Shapley expression in terms of simple sufficient statistics. Let $y_i$ denote the value of tuple $x_i$, let $A(W)=\sum_{x_i \in W} y_i$, and let $H_N=\sum_{r=1}^{N}1/r$ be the $N$-th harmonic number. For the game
\begin{equation}
    \nu_{AVG}(S)=
    \begin{cases}
        \frac{1}{|S|}\sum_{x_i \in S} y_i, & |S|>0,\\
        0, & |S|=0,
    \end{cases}
\end{equation}
the Shapley value of tuple $x_i$ for $N\geq 1$ (using the convention $H_1-1=0$) is
\begin{equation}
    \phi_i^{AVG}(W)
    =
    \frac{1}{N}
    \left[
        y_i +
        \frac{H_N-1}{N-1}
        \left(Ny_i - A(W)\right)
    \right].
    \label{eq:avg_tuple}
\end{equation}
This harmonic-number form is a direct consequence of the permutation definition (Appendix~\ref{app:avg}). It also serves as a baseline non-additive aggregate. The technical value here is that the tuple expression collapses to a predicate-level sufficient-statistics formula.

\textbf{Relation to prior work.} For a static database and a single-relation query with all facts endogenous, Eq.~\ref{eq:avg_tuple} was derived independently by Standke and Kimelfeld~\cite{standke2025tractability} as their Proposition~5.2. The cross-term sign in their printed statement is a typo. Their proof concludes with the expression equivalent to Eq.~\ref{eq:avg_tuple}, and the printed sign would violate Efficiency. Proposition~\ref{prop:avg} should therefore be read as a restatement of that tuple-level result in the windowed setting. The contribution of this section is the predicate-level form (Eq.~\ref{eq:avg_predicate}), its constant-size sufficient statistics, and its role as the base case of the maintenance framework of Sections~\ref{sec:maintenance}--\ref{sec:windows}.

\begin{proposition}[Exact AVG Attribution]
\label{prop:avg}
For $N \geq 2$, Eq.~\ref{eq:avg_tuple} is the exact Shapley value of $x_i$. For $N=1$, define the second term of Eq.~\ref{eq:avg_tuple} to be zero by the convention $H_1 - 1 = 0$ and the indeterminate $0/0$ resolved as $0$; then Eq.~\ref{eq:avg_tuple} yields $\phi_i = y_i$, which agrees with Efficiency. Eq.~\ref{eq:avg_predicate} below is the exact predicate attribution obtained by summing these tuple values over $C_{\mathcal{P}}$. The formulas hold verbatim for both count-based windows (with constant $N$) and time-based windows (with $N=N_t$, the active cardinality at evaluation time).
\end{proposition}
The proof follows directly from the permutation definition and is given in Appendix~\ref{app:avg}.

For a predicate $\mathcal{P}$, let $m_{\mathcal{P}}=|C_{\mathcal{P}}|$ and $A_{\mathcal{P}}=\sum_{x_i \in C_{\mathcal{P}}}y_i$. Summing Eq.~\ref{eq:avg_tuple} over all tuples satisfying the predicate gives
\begin{equation}
    \Phi_{\mathcal{P}}^{AVG}(W)
    =
    \frac{1}{N}
    \left[
        A_{\mathcal{P}}+
        \frac{H_N-1}{N-1}
        \left(NA_{\mathcal{P}}-m_{\mathcal{P}}A(W)\right)
    \right].
    \label{eq:avg_predicate}
\end{equation}
For a fixed-size count window, exact AVG attribution is therefore determined by four quantities: $N$, $A(W)$, $m_{\mathcal{P}}$, and $A_{\mathcal{P}}$. For time-based windows, $N$ in Eq.~\ref{eq:avg_predicate} is replaced by $N_t$ (Section~\ref{sec:windows}).

\begin{example}[AVG Attribution]
Consider $W=\{2,4,10\}$ and a predicate $\mathcal{P}$ selecting the tuple with value $10$. Here $N=3$, $A(W)=16$, $m_{\mathcal{P}}=1$, $A_{\mathcal{P}}=10$, and $H_3=11/6$. Eq.~\ref{eq:avg_predicate} gives
\begin{equation}
    \Phi_{\mathcal{P}}^{AVG}(W)
    =
    \frac{1}{3}
    \left[
        10+\frac{5/6}{2}(30-16)
    \right]
    =
    \frac{95}{18}\approx 5.28.
\end{equation}
By the Efficiency axiom, the three tuple-level Shapley values sum to the window average, $\nu_{AVG}(W)=16/3$. Direct evaluation of Eq.~\ref{eq:avg_tuple} gives
\begin{equation}
    \left(-\frac{13}{18},\ \frac{7}{9},\ \frac{95}{18}\right),
\end{equation}
which indeed sum to $16/3$.
The tuple with value $2$ receives a negative attribution, $-13/18$, because its presence lowers the average relative to many coalitions from which it is absent. The tuple with value $4$ lies closer to the mean and receives a smaller positive share.
In the latency-monitoring reading of the example, the predicate still receives most of the attribution because the selected tuple is well above the global mean.
\end{example}

\subsection{Interpretation: Share, Lift, and Scale}
\label{sec:scale}
Eq.~\ref{eq:avg_predicate} can be rewritten exactly in a form that makes its meaning clear. Let $\bar{y}_{\mathcal{P}} = A_{\mathcal{P}}/m_{\mathcal{P}}$ and $\bar{y}_W = A(W)/N$ be the predicate and window means. Then
\begin{equation}
\begin{aligned}
    \Phi_{\mathcal{P}}^{AVG}(W)
    &=
    \underbrace{\frac{m_{\mathcal{P}}}{N}\,\bar{y}_{\mathcal{P}}}_{\text{share}}
    \;+\;
    w(N)\,
    \underbrace{\frac{m_{\mathcal{P}}}{N}\big(\bar{y}_{\mathcal{P}}-\bar{y}_W\big)}_{\text{lift } \ell_{\mathcal{P}}},
\end{aligned}
    \label{eq:share_lift}
\end{equation}
with $w(N)=N(H_N-1)/(N-1)$. The Shapley value is therefore the group's share of the window mean plus $w(N)\approx H_N-1\approx\ln N + \gamma - 1$ times its lift, where $\gamma\approx 0.577$ is the Euler--Mascheroni constant. Three consequences follow.

First, for AVG the axioms do not produce a new signal. They fix a \emph{specific weighting} of two elementary statistics. Each statistic is a poor explanation on its own. Over any partition of $W$ the lifts sum to zero, so they describe deviation from the mean but not the value itself. The shares sum to $\bar{y}_W$ but ignore whether a group raises or lowers the aggregate. Eq.~\ref{eq:share_lift} is the one combination of the two that satisfies the Shapley axioms. Section~\ref{sec:eval:lift} measures how often this weighting changes the explanation compared to ranking by lift alone on real data.

Second, attributions depend on the window size. The weight $w(N)$ grows like $\ln N$. Thus, if the composition of the window stays fixed (same $m_{\mathcal{P}}/N$, $\bar{y}_{\mathcal{P}}$, $\bar{y}_W$), the attribution of a high-lift group keeps growing with $N$ even though the aggregate does not change. By Efficiency, other groups shrink by the same amount. Attributions of windows with very different sizes are therefore not directly comparable. This matters for time-based windows, where $N_t$ varies (Section~\ref{sec:windows}). The variance coefficients behave the same way (e.g., $C_B = H_N - H_N^{(2)}$ below). When comparing across windows, we recommend reporting the share and lift components next to $\Phi_{\mathcal{P}}$, or comparing only windows of similar size. The case study quantifies the effect on real data (Section~\ref{sec:eval:casestudy}).

Third, the decomposition keeps large attributions honest. A headline value such as ``$\Phi_{\mathcal{P}}=+\$63$ against an aggregate of $\$44$'' (Section~\ref{sec:eval:casestudy}) splits into a moderate share and a large amplified lift. Reporting both parts avoids over-reading the amplified one.

\subsection{Population and Sample Variance}
\label{sec:variance}
Variance is more challenging than AVG because it depends on both first and second moments. Nevertheless, both population and sample variance admit exact sufficient-statistics representations once the game is defined on all coalitions. Let $A(S)=\sum_{x_j \in S} y_j$ and $B(S)=\sum_{x_j \in S} y_j^2$.

\subsubsection{Population Variance}
Define
\begin{equation}
    \nu_{VAR\_POP}(S)=
    \begin{cases}
        \frac{B(S)}{|S|}-\left(\frac{A(S)}{|S|}\right)^2, & |S|>0,\\
        0, & |S|=0.
    \end{cases}
\end{equation}
For a fixed player $i$ and a random predecessor set $S$ of size $r$ drawn uniformly without replacement from $W\setminus\{x_i\}$, the exact Shapley value reduces, after eliminating empty coalitions, to
\begin{equation}
    \phi_i^{VAR}(W)
    =
    \frac{1}{N}
    \sum_{r=1}^{N-1}
    \mathbb{E}_{|S|=r}
    \left[
        \nu_{VAR\_POP}(S\cup\{x_i\})-\nu_{VAR\_POP}(S)
    \right].
    \label{eq:var_exact_target}
\end{equation}
The $r=0$ term contributes $\nu_{VAR\_POP}(\{x_i\})-\nu_{VAR\_POP}(\emptyset)=0$, because the variance of any singleton is zero and $\nu_{VAR\_POP}(\emptyset)=0$ by definition. For $r \geq 1$, let $n=N-1$, $T_i=A(W)-y_i$, and $U_i=B(W)-y_i^2$. Sampling $S$ without replacement from $W \setminus \{x_i\}$ gives the moment expressions (derived in Appendix~\ref{app:variance}):
\begin{align}
    \mathbb{E}[A(S)] &= \tfrac{r}{n}T_i,\\
    \mathbb{E}[B(S)] &= \tfrac{r}{n}U_i,\\
    \mathbb{E}[A(S)^2] &=
    \tfrac{r}{n}U_i+
    \tfrac{r(r-1)}{n(n-1)}(T_i^2-U_i),
    \label{eq:sample_sum_second_moment}
\end{align}
These expressions are valid for $N>2$. The $N=2$ case is handled separately in Appendix~\ref{app:variance} by enumerating the two possible orderings. Substituting these expectations into Eq.~\ref{eq:var_exact_target} and expanding step-by-step (Appendix~\ref{app:variance}) yields the closed form
\begin{equation}
    \phi_i^{VAR}(W)
    =
    \frac{1}{N}
    \sum_{r=1}^{N-1}
    D_i(r),
    \label{eq:var_tuple}
\end{equation}
where
\begin{equation}
\begin{aligned}
    D_i(r)
    &=
    -\frac{U_i}{n(r+1)}
    +
    \left(\frac{1}{r^2}-\frac{1}{(r+1)^2}\right)
    E_i^{(2)}(r) \\
    &\quad
    -
    \frac{2ry_iT_i}{n(r+1)^2}
    +
    \frac{r y_i^2}{(r+1)^2},
\end{aligned}
    \label{eq:var_delta}
\end{equation}
and $E_i^{(2)}(r)$ denotes the right-hand side of Eq.~\ref{eq:sample_sum_second_moment}.

The predicate-level form follows by summing Eq.~\ref{eq:var_delta} over all $x_i \in C_{\mathcal{P}}$. With $m_{\mathcal{P}}=|C_{\mathcal{P}}|$, $A_{\mathcal{P}}=\sum_{x_i \in C_{\mathcal{P}}} y_i$, and $B_{\mathcal{P}}=\sum_{x_i \in C_{\mathcal{P}}} y_i^2$, define
\begin{align}
    U_{\mathcal{P}} &= m_{\mathcal{P}}B(W)-B_{\mathcal{P}}, \label{eq:UP}\\
    YT_{\mathcal{P}} &= A(W)A_{\mathcal{P}}-B_{\mathcal{P}}, \label{eq:YTP}\\
    Q_{\mathcal{P}} &= m_{\mathcal{P}}A(W)^2-2A(W)A_{\mathcal{P}} \nonumber\\
    &\quad +2B_{\mathcal{P}}-m_{\mathcal{P}}B(W). \label{eq:QP}
\end{align}
These three quantities arise from the predicate-level sums
\begin{equation}
    \sum_{i \in C_{\mathcal{P}}} U_i,\qquad
    \sum_{i \in C_{\mathcal{P}}} y_i T_i,\qquad
    \sum_{i \in C_{\mathcal{P}}} (T_i^2 - U_i),
\end{equation}
respectively (Appendix~\ref{app:variance}). Summing over $r$ and collecting terms (Appendix~\ref{app:variance}) gives the exact predicate attribution for population variance:
\begin{equation}
    \Phi_{\mathcal{P}}^{VAR}(W)
    =
    \frac{1}{N}
    \left(
        C_U U_{\mathcal{P}}
        +
        C_Q Q_{\mathcal{P}}
        +
        C_{YT} YT_{\mathcal{P}}
        +
        C_B B_{\mathcal{P}}
    \right).
    \label{eq:var_predicate_constant}
\end{equation}
The coefficients arise in Appendix~\ref{app:variance} as finite sums over the coalition size $r$. Partial fractions and telescoping collapse every one of them to a closed form in the harmonic numbers $H_N$ and $H_N^{(2)}=\sum_{r=1}^{N}1/r^2$. Writing $G_N = H_N - H_N^{(2)}$,
\begin{align}
    C_B &= G_N,
    &
    C_{YT} &= -\frac{2\,G_N}{n},
    \label{eq:coeff_closed_1}\\
    C_U &= \frac{1-\tfrac{1}{N}-G_N}{n},
    &
    C_Q &= \frac{2\,G_N+\tfrac{1}{N}-1}{n(n-1)}.
    \label{eq:coeff_closed_2}
\end{align}
So exact population-variance attribution needs only $N$, $A(W)$, $B(W)$, $m_{\mathcal{P}}$, $A_{\mathcal{P}}$, and $B_{\mathcal{P}}$, plus $\mathcal{O}(1)$ arithmetic on $H_N$ and $H_N^{(2)}$. Both harmonic numbers update incrementally ($H_{N+1}=H_N+\frac{1}{N+1}$ and $H^{(2)}_{N+1}=H^{(2)}_N+\frac{1}{(N+1)^2}$). Thus no precomputation, no lookup table, and no $\mathcal{O}(N)$ work is needed at any cardinality. Time-based windows evaluate the same expressions at $N=N_t$ (Section~\ref{sec:windows}). We verified Eqs.~\ref{eq:coeff_closed_1}--\ref{eq:coeff_closed_2} against direct summation for all $N\le 10^4$.

\subsubsection{Sample Variance}
Database systems commonly expose sample variance as \texttt{VAR\_SAMP}. To make the Shapley game total on all coalitions, we set the singleton and empty values to zero:
\begin{equation}
    \nu_{VAR\_SAMP}(S)=
    \begin{cases}
        \frac{B(S)-A(S)^2/|S|}{|S|-1}, & |S|>1,\\
        0, & |S|\leq 1.
    \end{cases}
\end{equation}
This convention reflects the fact that no dispersion is observable from fewer than two tuples. It also avoids undefined game values in the Shapley sum.

The derivation follows the same steps as for population variance and is given in full in Appendix~\ref{app:variance}. Only the case $r=1$ needs separate treatment, because a one-tuple predecessor set has game value zero. The result uses the same summaries $U_{\mathcal{P}}$, $YT_{\mathcal{P}}$, and $Q_{\mathcal{P}}$:
\begin{equation}
    \Phi_{\mathcal{P}}^{VAR\_SAMP}(W)
    =
    \frac{1}{N}
    \left(
        C_U^{S} U_{\mathcal{P}}
        +
        C_Q^{S} Q_{\mathcal{P}}
        +
        C_{YT}^{S} YT_{\mathcal{P}}
        +
        C_B^{S} B_{\mathcal{P}}
    \right),
    \label{eq:varsamp_predicate_constant}
\end{equation}
with coefficients that collapse even further than in the population case, to expressions in $H_N$ alone:
\begin{align}
    C^S_U &= \frac{2-H_N}{n},
    &
    C^S_{YT} &= \frac{2\,(1-H_N)}{n},
    \label{eq:coeff_closed_s1}\\
    C^S_B &= H_N-1,
    &
    C^S_Q &= \frac{2\big(H_N-\tfrac{3}{2}\big)}{n(n-1)},
    \label{eq:coeff_closed_s2}
\end{align}
Thus sample variance inherits the same $\mathcal{O}(1)$ incremental coefficient maintenance.

\begin{proposition}[Variance]
\label{prop:variance}
Under the coalition-value definitions above, Eq.~\ref{eq:var_predicate_constant} and Eq.~\ref{eq:varsamp_predicate_constant} give exact predicate-level Shapley attribution for population and sample variance. The formulas hold for count-based windows and for time-based windows after substituting $N=N_t$.
\end{proposition}
The proof is given in Appendix~\ref{app:variance}.

\textbf{Empirical Validation.} Section~\ref{sec:eval:exactness} verifies the closed forms (Eq.~\ref{eq:avg_tuple}, Eq.~\ref{eq:var_predicate_constant}, and Eq.~\ref{eq:varsamp_predicate_constant}) against exhaustive coalition enumeration over $10{,}010$ random window configurations. The observed deviations are at double-precision floating-point level.

\textbf{Sign interpretation.} Variance attributions can be negative. A tuple or predicate close to the window mean tends to reduce dispersion relative to coalitions where it is absent. An extreme tuple or predicate tends to increase dispersion. Thus $\Phi_{\mathcal{P}}^{VAR}<0$ should be read as a stabilizing contribution to the variance aggregate, not as an error.

\begin{example}[Variance Attribution]
For the same window $W=\{2,4,10\}$, the population variance is $104/9$. The tuple-level population-variance Shapley values are $(3.685,\,0.185,\,7.685)$, summing to $104/9$. A predicate selecting value $10$ receives positive variance attribution because it is the most dispersion-increasing tuple. Under sample variance with the singleton convention above, the Shapley values are $(5.444,\,-1.556,\,13.444)$, summing to $17.333 \approx 52/3 = \nu_{VAR\_SAMP}(W)$. The middle tuple receives negative attribution because it tends to pull coalitions toward the mean.
\end{example}

\subsection{Linearity in Predicate Summaries}
\label{sec:linearity}
The structural property that enables predicate-level evaluation is summarized by the following lemma.

\begin{lemma}[Affine Linearity]
\label{lem:linearity}
For SUM, COUNT, AVG, population variance, and sample variance, once $N$, $A(W)$, and $B(W)$ are fixed, the predicate-level Shapley attribution $\Phi_{\mathcal{P}}(W)$ is an affine function of the predicate summary triple $(m_{\mathcal{P}}, A_{\mathcal{P}}, B_{\mathcal{P}})$.
\end{lemma}
\begin{proof}
For SUM, $\Phi_{\mathcal{P}}^{SUM}(W) = A_{\mathcal{P}}$. For COUNT, $\Phi_{\mathcal{P}}^{COUNT}(W) = m_{\mathcal{P}}$. For AVG, Eq.~\ref{eq:avg_predicate} is affine in $(m_{\mathcal{P}}, A_{\mathcal{P}})$. For population and sample variance, the auxiliary quantities $U_{\mathcal{P}}$, $YT_{\mathcal{P}}$, $Q_{\mathcal{P}}$ defined in Eq.~\ref{eq:UP}--\ref{eq:QP} are themselves linear in $(m_{\mathcal{P}}, A_{\mathcal{P}}, B_{\mathcal{P}})$ once $A(W)$ and $B(W)$ are fixed; substituting into Eq.~\ref{eq:var_predicate_constant} or Eq.~\ref{eq:varsamp_predicate_constant} preserves the affine property.
\end{proof}

Given the tuple-level closed forms, Lemma~\ref{lem:linearity} is immediate. Each $\phi_i$ is affine in $(1, y_i, y_i^2)$ with coefficients fixed by $(N, A(W), B(W))$, and summation over $C_{\mathcal{P}}$ preserves affinity. The technical content of this section therefore lies in the tuple-level derivations (Eq.~\ref{eq:avg_tuple} and Appendix~\ref{app:variance}) and in the coefficient closed forms. The lemma itself is the simple but load-bearing observation that converts them into predicate-level answers without ever materializing tuple-level Shapley values. It also enables the additivity-across-subsets corollary used by pane-based maintenance frameworks (Section~\ref{sec:additivity}). Finally, it makes sampling-based approximation of $\Phi_{\mathcal{P}}$ for ad-hoc predicates well posed as estimation of an affine functional of maintained summaries.

\subsection{Summary Maintenance Under Slides}
\label{sec:maintenance}
The pane-level summaries for a predicate $\mathcal{P}$ are updated by the window delta. For a count-based slide $\Delta=(x_{out},x_{in})$ and numeric attribute $y$:
\begin{align}
    m_{\mathcal{P}} &\leftarrow m_{\mathcal{P}} + \ind_{\mathcal{P}}(x_{in})-\ind_{\mathcal{P}}(x_{out}),
    \label{eq:update_m}\\
    A_{\mathcal{P}} &\leftarrow A_{\mathcal{P}} + \ind_{\mathcal{P}}(x_{in})y_{in}
    \nonumber\\
    &\qquad - \ind_{\mathcal{P}}(x_{out})y_{out},
    \label{eq:update_A}\\
    B_{\mathcal{P}} &\leftarrow B_{\mathcal{P}} + \ind_{\mathcal{P}}(x_{in})y_{in}^2
    \nonumber\\
    &\qquad - \ind_{\mathcal{P}}(x_{out})y_{out}^2.
    \label{eq:update_B}
\end{align}
For a multi-tuple time-based delta $\Delta_t=(X^{out}_t,X^{in}_t)$ the same updates apply tuple-by-tuple over both sets, so the cost is $\mathcal{O}(|\Delta_t|)$ per affected predicate. Global summaries $N$ (or $N_t$), $A(W)$, and $B(W)$ are updated analogously.

\paragraph{Numerical stability.}
Long-running streams expose two floating-point issues. The first is cancellation in the variance formulas. Terms such as $A(W)A_{\mathcal{P}}-B_{\mathcal{P}}$ and the $m_{\mathcal{P}}A(W)^2$ part of $Q_{\mathcal{P}}$ subtract nearly equal quantities. When the mean is much larger than the spread, most significant digits cancel and the result is dominated by rounding noise. This is the same failure mode as classical one-pass variance. A standard remedy is to center the values first, $y_i \leftarrow y_i - c$ for a coarse location estimate $c$. AVG attribution is recovered by adding back $c\,m_{\mathcal{P}}/N$, and VAR attribution is invariant to the shift. The second issue is drift from repeated updates. Adding and subtracting millions of increments into the same summaries slowly loses low-order bits. Pane-based maintenance (Corollary~\ref{cor:panes}) avoids most of this for free. Each pane keeps its own partial sums, and the window summary is rebuilt by adding a few pane totals rather than a long history of tiny updates. Compensated summation (which carries a small correction for bits lost at each add) can shrink whatever drift remains. Empirically, the hard case is the one that needs the shift. With $N=10^5$, values $\sim\mathcal{N}(10^6,1)$, and $10^6$ slides of plain additive updates, unshifted variance attribution drifted to $23\%$ relative error. The same run with $c=10^6$ stayed accurate to $7\times10^{-13}$. On our month-long taxi case study, where the fare mean is comparable to its spread (Section~\ref{sec:eval:casestudy}), ordinary double-precision updates already kept the Efficiency identity within $10^{-10}$ relative deviation. No special mitigation was required.

\begin{algorithm}[t]
\caption{Per-slide summary maintenance}
\label{alg:summary-maintenance}
\begin{algorithmic}[1]
\State \textbf{Input:} delta $\Delta_t=(X^{out}_t,X^{in}_t)$ (for count-based slides, take $X^{out}_t=\{x_{out}\}$ and $X^{in}_t=\{x_{in}\}$), set $\mathcal{R}$ of predicates whose summaries are maintained
\State Update global summaries $N$ (or $N_t$), $A(W)$, $B(W)$ over all tuples in $X^{out}_t \cup X^{in}_t$
\ForAll{tuples $x \in X^{out}_t \cup X^{in}_t$}
    \ForAll{$\mathcal{P} \in \mathcal{R}$ with $\mathcal{P}(x)=1$}
        \State Update $m_{\mathcal{P}}$, $A_{\mathcal{P}}$, $B_{\mathcal{P}}$ using Eqs.~\ref{eq:update_m}--\ref{eq:update_B}
    \EndFor
\EndFor
\If{$N$ (or $N_t$) changed since last evaluation}
    \State Update running $H_N$, $H^{(2)}_N$ (or $H_{N_t}$, $H^{(2)}_{N_t}$) and evaluate the coefficient closed forms Eqs.~\ref{eq:coeff_closed_1}--\ref{eq:coeff_closed_s2} \Comment{$\mathcal{O}(1)$}
\EndIf
\State On query, evaluate the requested closed-form attribution
\end{algorithmic}
\end{algorithm}

Combined with Lemma~\ref{lem:linearity}, the additive update rules above give the following corollary.

\begin{corollary}[Per-Slide Cost]
\label{cor:perslide}
For SUM, COUNT, AVG, population variance, and sample variance, the predicate attribution $\Phi_{\mathcal{P}}(W)$ can be evaluated in $\mathcal{O}(1)$ from the maintained summaries and the running harmonic numbers $H_N$, $H^{(2)}_N$ (Eqs.~\ref{eq:coeff_closed_1}--\ref{eq:coeff_closed_s2}). The per-slide summary update cost is $\mathcal{O}(1)$ per affected predicate for count-based windows and $\mathcal{O}(|\Delta_t|)$ for time-based windows; no coefficient precomputation or caching is required.
\end{corollary}

Table~\ref{tab:complexity} contrasts this with naive recomputation per window.

\begin{table}[t]
\small
\setlength{\tabcolsep}{3pt}
\renewcommand{\arraystretch}{0.95}
\centering
\begin{tabular}{|l|c|c|}
\hline
\textbf{Aggregate Type} & \textbf{Recompute per Window} & \textbf{Closed-Form Maintenance} \\ \hline
\multicolumn{3}{|c|}{\textit{Count-based windows (constant $N$)}} \\ \hline
SUM/COUNT & $\mathcal{O}(N)$ & $\mathcal{O}(1)$ exact update \\ \hline
AVG & $\mathcal{O}(N)$ & $\mathcal{O}(1)$ exact update \\ \hline
VAR\_POP/VAR\_SAMP & $\mathcal{O}(N)$ & $\mathcal{O}(1)$ exact update \\ \hline
\multicolumn{3}{|c|}{\textit{Time-based windows (variable $N_t$, delta size $|\Delta_t|$)}} \\ \hline
SUM/COUNT & $\mathcal{O}(N_t)$ & $\mathcal{O}(|\Delta_t|)$ exact update \\ \hline
AVG & $\mathcal{O}(N_t)$ & $\mathcal{O}(|\Delta_t|)$ exact update \\ \hline
VAR\_POP/VAR\_SAMP & $\mathcal{O}(N_t)$ & $\mathcal{O}(|\Delta_t|)$ exact update \\ \hline
\end{tabular}
\caption{Exact maintenance cost per affected predicate. Constant-time entries assume (i)~the predicate summaries are maintained and (ii)~that tuple-membership tests are $\mathcal{O}(1)$. All attribution coefficients are closed forms in the running harmonic numbers $H_N$ and $H^{(2)}_N$ (Eqs.~\ref{eq:coeff_closed_1}--\ref{eq:coeff_closed_s2}). Coefficient access is therefore $\mathcal{O}(1)$ for both window types with no precomputation.}
\label{tab:complexity}
\end{table}

\subsection{Additivity Across Disjoint Subsets}
\label{sec:additivity}
The summaries $A(W)$, $B(W)$, $m_{\mathcal{P}}$, $A_{\mathcal{P}}$, $B_{\mathcal{P}}$ are all sums of per-tuple quantities. They are therefore additive over disjoint partitions of $W$. The following corollary captures this systems-relevant consequence.

\begin{corollary}[Pane Compatibility]
\label{cor:panes}
Let $\{P_1,\dots,P_k\}$ be any partition of $W$. Then
\begin{equation}
    A(W)=\sum_{j=1}^{k}A(P_j), \quad
    B(W)=\sum_{j=1}^{k}B(P_j),
\end{equation}
and analogously for $m_{\mathcal{P}}$, $A_{\mathcal{P}}$, $B_{\mathcal{P}}$. Consequently, the closed forms Eq.~\ref{eq:avg_predicate}, Eq.~\ref{eq:var_predicate_constant}, and Eq.~\ref{eq:varsamp_predicate_constant} can be evaluated from per-pane summaries via additive merging, without merging local Shapley vectors or applying any global correction.
\end{corollary}

This makes the closed forms drop-in compatible with existing pane-based sliding-window aggregation frameworks~\cite{li2005pane,tangwongsan2015general}. The Shapley statistics are added to the standard set of per-pane aggregates, and the closed-form formula is evaluated once after pane merging.

\paragraph{Toward a characterization.}
All aggregates treated here share one property. Once a few window-level statistics are fixed, the expected Shapley marginal contribution of a tuple, or of a predicate coalition, depends only on its low-order moments. Section~\ref{sec:boundary} makes this precise in both directions. It gives a sufficiency theorem covering every \emph{moment-polynomial} game and an impossibility result showing that order statistics escape every fixed moment order.

\section{Compositional Predicates via Atomic Refinement}
\label{sec:overlap}

When the system maintains $K$ predicates $\mathcal{P}_1,\dots,\mathcal{P}_K$, two phenomena complicate the per-predicate maintenance described above. First, predicates may \emph{overlap}: a tuple can satisfy several at once. The family $\{\Phi_{\mathcal{P}_k}\}_{k=1}^{K}$ then counts those tuples multiple times and does not sum to $\nu(W)-\nu(\emptyset)$. Second, users frequently issue \emph{compositional} queries (intersections, unions, complements, set-differences). For these queries, a single per-predicate summary is insufficient. We address both by maintaining sufficient statistics on the \emph{atomic refinement} of the predicates.

\subsection{Signatures and Atoms}
Assign to each active tuple $x$ a Boolean signature
\begin{equation}
    \mathrm{sig}(x) = (\mathcal{P}_1(x),\dots,\mathcal{P}_K(x)) \in \{0,1\}^K.
    \label{eq:signature}
\end{equation}
For each $s\in\{0,1\}^K$, define the \emph{atom} $\alpha_s = \{x \in W \mid \mathrm{sig}(x)=s\}$. The family $\{\alpha_s\}_{s \in \{0,1\}^K}$ partitions $W$. Let $\mathcal{S}_W = \{s : \alpha_s \neq \emptyset\}$ be the set of \emph{active} signatures; $|\mathcal{S}_W| \leq \min(2^K, N)$. For each active signature $s$, we maintain
\begin{equation}
    m_{\alpha_s} = |\alpha_s|,\quad
    A_{\alpha_s} = \sum_{x_i \in \alpha_s} y_i,\quad
    B_{\alpha_s} = \sum_{x_i \in \alpha_s} y_i^2.
\end{equation}
Each atom is itself a coalition of active tuples. The closed forms Eq.~\ref{eq:avg_predicate}, Eq.~\ref{eq:var_predicate_constant}, and Eq.~\ref{eq:varsamp_predicate_constant} therefore apply with $(m_{\mathcal{P}},A_{\mathcal{P}},B_{\mathcal{P}})$ replaced by $(m_{\alpha_s},A_{\alpha_s},B_{\alpha_s})$. We denote the resulting atom attribution by $\Phi_{\alpha_s}(W)$.

\subsection{Compositional Queries}
For any Boolean combination $\mathcal{Q}$ of $\mathcal{P}_1,\dots,\mathcal{P}_K$, let $\mathcal{S}(\mathcal{Q})=\{s\in\{0,1\}^K : \mathcal{Q}(s)=1\}$, where $\mathcal{Q}(s)$ denotes the truth value of $\mathcal{Q}$ on a tuple with signature $s$. The coalition $C_{\mathcal{Q}}$ is then a disjoint union of atoms, and the attribution decomposes additively:
\begin{equation}
    \Phi_{\mathcal{Q}}(W) = \sum_{x_i \in C_{\mathcal{Q}}} \phi_i(\nu, W) = \sum_{s \in \mathcal{S}(\mathcal{Q})\,\cap\,\mathcal{S}_W} \Phi_{\alpha_s}(W).
    \label{eq:compositional_query}
\end{equation}
Standard cases include
\begin{itemize}
    \item \emph{Single predicate:} $\Phi_{\mathcal{P}_k}=\sum_{s:\,s_k=1}\Phi_{\alpha_s}$.
    \item \emph{Intersection:} $\Phi_{\mathcal{P}_j\wedge\mathcal{P}_k}=\sum_{s:\,s_j=s_k=1}\Phi_{\alpha_s}$.
    \item \emph{Complement:} $\Phi_{\neg\mathcal{P}_k}=\sum_{s:\,s_k=0}\Phi_{\alpha_s}$.
    \item \emph{Set-difference:} $\Phi_{\mathcal{P}_j\wedge\neg\mathcal{P}_k}=\sum_{s:\,s_j=1,\,s_k=0}\Phi_{\alpha_s}$.
    \item \emph{Union (inclusion--exclusion):} $\Phi_{\mathcal{P}_j\vee\mathcal{P}_k}=\Phi_{\mathcal{P}_j}+\Phi_{\mathcal{P}_k}-\Phi_{\mathcal{P}_j\wedge\mathcal{P}_k}$, equivalently $\sum_{s:\,s_j\vee s_k=1}\Phi_{\alpha_s}$.
\end{itemize}
Answering a compositional query costs one closed-form evaluation per active atom satisfying $\mathcal{Q}$, i.e., $\mathcal{O}(|\mathcal{S}(\mathcal{Q})\cap\mathcal{S}_W|)$ time. This cost is at most $\mathcal{O}(\min(2^K,N))$. In practice, it tracks the number of signatures that actually occur (Section~\ref{sec:eval:scaling}). Alternatively, summaries $(m,A,B)$ of the participating atoms can be added first and the closed form evaluated once on the merged summary, by Corollary~\ref{cor:panes}.

\subsection{Restoration of Efficiency}
Because $\{\alpha_s\}_{s\in\mathcal{S}_W}$ partitions $W$,
\begin{equation}
    \sum_{s\in\mathcal{S}_W} \Phi_{\alpha_s}(W) = \nu(W) - \nu(\emptyset).
    \label{eq:atom_efficiency}
\end{equation}
For any partition $\{\mathcal{Q}_1,\dots,\mathcal{Q}_M\}$ of $\{0,1\}^K$ into Boolean combinations, the family of compositional attributions inherits the same partition property: $\sum_m \Phi_{\mathcal{Q}_m}=\nu(W)-\nu(\emptyset)$. In particular, $\Phi_{\mathcal{P}}+\Phi_{\neg\mathcal{P}}=\nu(W)-\nu(\emptyset)$ for any $\mathcal{P}$. Under overlap, $\sum_k \Phi_{\mathcal{P}_k}\neq \nu(W)-\nu(\emptyset)$ is not a defect of the closed forms. It is a consequence of double-counting tuples that satisfy multiple predicates. The atom decomposition makes this phenomenon transparent and corrects it.

\subsection{Maintenance Under Slides}
Atom summaries are updated by the same delta rules as predicate summaries (Eqs.~\ref{eq:update_m}--\ref{eq:update_B}), applied to the atom selected by the tuple's signature. For each $x \in X^{out}_t \cup X^{in}_t$, the engine evaluates the $K$ predicates, forms $\mathrm{sig}(x)$ in $\mathcal{O}(K)$ time, and updates the summaries of the corresponding atom. New signatures are allocated lazily, and atoms whose count drops to zero are reclaimed. Per-tick maintenance cost is therefore $\mathcal{O}(|\Delta_t|\,K)$ across all atoms, and the state is $\mathcal{O}(|\mathcal{S}_W|) \leq \mathcal{O}(\min(2^K, N))$.
In practice, the $K$ predicate tests are compiled into a $K$-bit mask per tuple. Signature evaluation then behaves like any vectorized filter operator rather than $K$ scattered branches.

\subsection{State Growth in Practice}
The $2^K$ upper bound on $|\mathcal{S}_W|$ is rarely tight. Categorical attributes with mutually exclusive values (e.g., \texttt{region='EU'} vs.\ \texttt{region='US'}) allow at most one of the corresponding bits to be set, so impossible signatures never appear. When predicates target disjoint dimensions (region, tier, service), $|\mathcal{S}_W|$ grows with the product of the per-dimension cardinalities rather than with $2^K$. Lazy allocation keeps only the signatures that actually occur. If the exact atom state still grows too large for memory, bounded-memory fallbacks apply. These include sketching or heavy-hitter tracking on signatures, partial materialization of atoms, or maintaining summaries only for explicitly enumerated intersections and answering the remaining compositional queries by inclusion--exclusion over those.

\begin{example}[Compositional Predicate Attribution]
For $K=2$ with $\mathcal{P}_1=(\text{region}=\text{EU})$ and $\mathcal{P}_2=(\text{tier}=\text{premium})$, the four signatures yield atoms $\alpha_{00}$ (US, basic), $\alpha_{01}$ (US, premium), $\alpha_{10}$ (EU, basic), $\alpha_{11}$ (EU, premium). From their summaries,
``EU contribution'' $=\Phi_{\alpha_{10}}+\Phi_{\alpha_{11}}$;
``premium contribution'' $=\Phi_{\alpha_{01}}+\Phi_{\alpha_{11}}$;
``EU and premium'' $=\Phi_{\alpha_{11}}$;
``EU but not premium'' $=\Phi_{\alpha_{10}}$;
``EU or premium'' $=\Phi_{\alpha_{01}}+\Phi_{\alpha_{10}}+\Phi_{\alpha_{11}}$;
all answered exactly, with $\Phi_{\alpha_{00}}+\Phi_{\alpha_{01}}+\Phi_{\alpha_{10}}+\Phi_{\alpha_{11}}=\nu(W)-\nu(\emptyset)$.
\end{example}

\begin{proposition}[Compositional Predicate Attribution]
\label{prop:compositional}
Fix $K$ predicates $\mathcal{P}_1,\dots,\mathcal{P}_K$ and let $\{\alpha_s\}_{s\in\mathcal{S}_W}$ be their atomic refinement on the active window. For any Boolean combination $\mathcal{Q}$ of the predicates, the predicate attribution $\Phi_{\mathcal{Q}}(W)$ is given exactly by Eq.~\ref{eq:compositional_query}, with each atom attribution $\Phi_{\alpha_s}(W)$ computed from $(m_{\alpha_s},A_{\alpha_s},B_{\alpha_s})$ and the global summaries via Eq.~\ref{eq:avg_predicate}, Eq.~\ref{eq:var_predicate_constant}, or Eq.~\ref{eq:varsamp_predicate_constant}. Per-tick maintenance is $\mathcal{O}(|\Delta_t|\,K)$, and the state is $\mathcal{O}(\min(2^K, N))$. The atom decomposition restores Shapley efficiency over the window: $\sum_{s\in\mathcal{S}_W}\Phi_{\alpha_s}(W)=\nu(W)-\nu(\emptyset)$.
\end{proposition}
\begin{proof}
Atoms partition $W$ by construction, so $C_{\mathcal{Q}}$ is a disjoint union of atoms and $\Phi_{\mathcal{Q}}=\sum_{x_i\in C_{\mathcal{Q}}}\phi_i=\sum_s\Phi_{\alpha_s}$. Each closed form was derived for an arbitrary subset of the window using only that subset's $(m,A,B)$ and the global summaries, and atoms are such subsets. The maintenance cost holds because each slide tuple needs $K$ predicate evaluations to form its signature and one atom-summary update. Efficiency follows from the Efficiency axiom applied to the partition $\{\alpha_s\}$.
\end{proof}

\section{Time-Based Windows}
\label{sec:windows}
The closed forms Eq.~\ref{eq:avg_tuple}, Eq.~\ref{eq:avg_predicate}, Eq.~\ref{eq:var_predicate_constant}, and Eq.~\ref{eq:varsamp_predicate_constant} were derived without using the count-based fixed-cardinality assumption. They depend on the window only through its cardinality $N$ and the additive moments $A(W)$, $B(W)$ (and the corresponding predicate or atom summaries). For time-based windows we therefore evaluate the same formulas with $N$ replaced by the current active cardinality $N_t$. Two practical concerns arise that did not arise in the count-based case: (i) the per-tick delta is multi-tuple, and (ii) the $N$-dependent coefficients must track $N_t$. By the closed forms of Section~\ref{sec:variance}, the latter reduces to maintaining the two running harmonic numbers.

\subsection{Multi-Tuple Deltas}
For a time-based slide $\Delta_t=(X^{out}_t,X^{in}_t)$, the per-tick maintenance cost is $\mathcal{O}(|\Delta_t|)$ for each affected predicate, because all summary updates are additive (Algorithm~\ref{alg:summary-maintenance}). Under stationary arrivals with rate $\lambda$ and tick interval $h$, the expected delta size is $\mathbb{E}[|\Delta_t|]=2\lambda h$ in steady state, because the expected insert and expiration rates match. Thus, for a fixed emit interval, the expected update work is controlled by the number of arrivals and expirations in that interval rather than by the full window cardinality.

\subsection{\texorpdfstring{$N$}{N}-Dependent Coefficient Updates}
The coefficients depend on the window only through $H_{N_t}$ and $H^{(2)}_{N_t}$ (Eqs.~\ref{eq:coeff_closed_1}--\ref{eq:coeff_closed_s2}). Both update incrementally,
\begin{equation}
    H_{N+1}=H_N+\frac{1}{N+1},
    \qquad
    H^{(2)}_{N+1}=H^{(2)}_N+\frac{1}{(N+1)^2},
\end{equation}
When $N_t$ jumps by more than one within a tick, extending the pair by $|\Delta N_t|$ steps costs $\mathcal{O}(|\Delta N_t|) \le \mathcal{O}(|\Delta_t|)$. This cost is already dominated by the summary updates of the same tick. In practice, one keeps the prefix arrays $H_{1},\dots,H_{N^{\mathrm{seen}}}$ and $H^{(2)}_{1},\dots,H^{(2)}_{N^{\mathrm{seen}}}$, extended lazily to the largest cardinality seen. Coefficient maintenance is therefore free in both asymptotic and practical terms, for any pattern of cardinality variation, with no precomputation and no assumptions on $N_{\max}$.

\subsection{Edge Cases and Out-of-Order Arrivals}
\paragraph{Edge cases.}
When $N_t=0$ the window is empty and all attributions are undefined; we emit a sentinel and continue. When $N_t=1$, by Efficiency $\phi_i=\nu(W)$. The AVG closed form (Eq.~\ref{eq:avg_tuple}) reduces correctly because $H_1-1=0$, giving $\phi_i=y_i/1=y_i$. The variance games are zero on singletons by definition. When $N_t$ shrinks across a tick because more tuples expire than arrive, the same updates apply with the corresponding $\delta<0$.

\paragraph{Out-of-order arrivals.}
The analysis above assumes that within a tick all tuple insertions and expirations are determined unambiguously by event-time watermarks. Late arrivals, meaning tuples whose timestamp falls inside an already-emitted window, require either watermark-based reordering of the stream or retroactive correction of the affine attribution form. We treat full out-of-order support as future work. In practice, watermark policies in Flink and similar engines bound the out-of-order delay so that retroactive corrections, if needed, are applied to a small constant number of recent emit ticks.

\paragraph{Retractions and corrections.}
Engines that model corrections with retractions (e.g., an explicit delete or negated update followed by an inserted replacement) pair cleanly with the maintenance story in this paper. Predicate and global summaries are affine in the per-tuple increments in Eqs.~\ref{eq:update_m}--\ref{eq:update_B}. Cancelling an erroneous tuple therefore reapplies the same updates with opposite sign, and inserting the corrected tuple applies them forward again.

\begin{proposition}[Time-Based Maintenance]
\label{prop:timebased}
For a time-based window with stochastic cardinality $N_t$ and per-tick delta $\Delta_t$, the predicate attributions Eq.~\ref{eq:avg_predicate}, Eq.~\ref{eq:var_predicate_constant}, and Eq.~\ref{eq:varsamp_predicate_constant} are exact at every emit tick once $N_t$ is substituted for $N$. The per-tick cost is $\mathcal{O}(|\Delta_t|)$ summary updates; coefficient maintenance is $\mathcal{O}(1)$ amortized via the harmonic-number closed forms Eqs.~\ref{eq:coeff_closed_1}--\ref{eq:coeff_closed_s2}.
\end{proposition}
\begin{proof}
The closed forms use only the cardinality and the additive moments of the current window, so substituting $N_t$ preserves the derivations. The cost bounds follow from Algorithm~\ref{alg:summary-maintenance} applied to multi-tuple deltas and from the incremental harmonic-number updates of Section~\ref{sec:windows}.
\end{proof}

\section{The Boundary of Sufficient-Statistic Attribution}
\label{sec:boundary}

The five aggregates of Section~\ref{sec:closedforms} are instances of one structural phenomenon. This section describes its boundary from both sides. A sufficiency theorem identifies a broad class of games whose predicate attribution is affine in maintained power sums. An impossibility result shows that order statistics fall outside every fixed moment order.

\subsection{Sufficiency: Moment-Polynomial Games}
For a coalition $S \subseteq W$ and $\ell \geq 1$, let $M_\ell(S)=\sum_{x_j \in S} y_j^{\ell}$ denote the $\ell$-th power sum, with $M_\ell(\emptyset)=0$. Note $M_1=A$ and $M_2=B$ in the earlier notation.

\begin{definition}[Moment-polynomial game]
\label{def:mompoly}
Fix integers $d \geq 1$ and $p \geq 0$. A window game $\nu$ is a \emph{moment-polynomial game of order $(d,p)$} if there exist coefficient functions $g_\kappa:\{0,1,\dots,N\}\to\mathbb{R}$, indexed by multi-indices $\kappa=(\kappa_1,\dots,\kappa_d)\in\mathbb{N}^d$ with $|\kappa|=\sum_\ell \kappa_\ell \leq p$, such that for every $S \subseteq W$,
\begin{equation}
    \nu(S) \;=\; \sum_{|\kappa|\leq p} g_\kappa(|S|)\prod_{\ell=1}^{d} M_\ell(S)^{\kappa_\ell}.
    \label{eq:mompoly}
\end{equation}
\end{definition}
All five aggregates of Section~\ref{sec:closedforms} are moment-polynomial: SUM is $M_1(S)$ (order $(1,1)$); COUNT is $g_0(|S|)=|S|$ (order $(1,0)$); AVG is $g(|S|)M_1(S)$ with $g(s)=1/s$ for $s>0$ (order $(1,1)$); population variance is $\frac{1}{|S|}M_2(S)-\frac{1}{|S|^2}M_1(S)^2$ (order $(2,2)$); sample variance is analogous with $1/(|S|-1)$ factors. The class also contains, e.g., the $q$-th central-moment numerator $\sum_{x_j\in S}(y_j-\overline{y}_S)^q$ for any fixed $q$, and sums of fixed polynomials of the values.

\begin{theorem}[Sufficiency]
\label{thm:sufficiency}
Let $\nu$ be a moment-polynomial game of order $(d,p)$ and let $D=dp$. Then:
\begin{enumerate}
    \item[(i)] There exist coefficients $c_0,\dots,c_D$, each a polynomial in the global power sums $M_1(W),\dots,M_D(W)$ whose coefficients depend only on $N$ (and on the $g_\kappa$), such that $\phi_i(\nu,W)=\sum_{\ell=0}^{D} c_\ell\, y_i^{\ell}$ for every $x_i \in W$.
    \item[(ii)] For every predicate $\mathcal{P}$,
    $\Phi_{\mathcal{P}}(W)=\sum_{\ell=0}^{D} c_\ell\, M_{\ell,\mathcal{P}}$,
    where $M_{0,\mathcal{P}}=m_{\mathcal{P}}$ and $M_{\ell,\mathcal{P}}=\sum_{x_i \in C_{\mathcal{P}}} y_i^{\ell}$; that is, $\Phi_{\mathcal{P}}$ is affine in the predicate power-sum vector, generalizing Lemma~\ref{lem:linearity}.
    \item[(iii)] All required summaries are additive over disjoint subsets. Consequently, for fixed $(d,p)$, Corollaries~\ref{cor:perslide} and~\ref{cor:panes} extend verbatim with $\mathcal{O}(D)$ summaries per predicate or atom: per-slide maintenance is $\mathcal{O}(D)$ per affected predicate and queries evaluate in $\mathcal{O}(D)$ given coefficients that depend only on $N$ and the global power sums.
\end{enumerate}
\end{theorem}
The proof (Appendix~\ref{app:boundary}) rests on two classical facts. First, over a uniformly random $r$-subset drawn without replacement, the probability that $q$ fixed distinct tuples are jointly included is the falling-factorial ratio $(r)_q/(n)_q$. Second, products of power sums expand into sums over distinct-index tuples that are again polynomials in power sums. The degree bound $D=dp$ is not tight for specific games. For the variance games, cancellations reduce the required predicate summaries to $(m_{\mathcal{P}},A_{\mathcal{P}},B_{\mathcal{P}})$ as in Section~\ref{sec:variance}.

\begin{remark}[Multivariate extension]
\label{rem:multivariate}
If tuples carry a value vector $(y_j,z_j,\dots)$ and Eq.~\ref{eq:mompoly} is stated over mixed power sums $M_{a,b}(S)=\sum_{x_j\in S}y_j^a z_j^b$, the proof carries over unchanged. This covers, e.g., the covariance game $\nu_{COV}(S)=\frac{1}{|S|}M_{1,1}(S)-\frac{1}{|S|^2}M_{1,0}(S)M_{0,1}(S)$. Its predicate attribution is therefore affine in the predicate summaries $(m_{\mathcal{P}},\sum y_i,\sum z_i,\sum y_i z_i,\sum y_i^2,\sum z_i^2)$.
\end{remark}

\subsection{Impossibility: Order Statistics}
The MAX game is $\nu_{MAX}(S)=\max_{x_j \in S} y_j$ for $S\neq\emptyset$ and $\nu_{MAX}(\emptyset)=0$.

\begin{proposition}[No finite-moment representation for order statistics]
\label{prop:impossibility}
For every $d \geq 1$ there exist windows $W$, $W'$ with $|W|=|W'|=N$ and $M_\ell(W)=M_\ell(W')$ for all $\ell \leq d$, but $\nu_{MAX}(W)\neq\nu_{MAX}(W')$. Consequently, no attribution rule computed from $(N, M_1(W),\dots,M_d(W))$ and per-predicate power sums up to order $d$ can satisfy the Efficiency axiom for the MAX game: the totals it must reproduce already differ on inputs it cannot distinguish. The same holds for MIN and for any fixed order statistic, hence for quantiles and the median.
\end{proposition}
For $d=2$ an explicit witness is $W=\{1,5,6\}$ and $W'=\{2,3,7\}$: both have $N=3$, $M_1=12$, and $M_2=62$, yet $\max W = 6 \neq 7 = \max W'$. The general construction (Appendix~\ref{app:boundary}) perturbs a window with distinct values along the kernel of the moment map, which has dimension $\geq 2$ whenever $N \geq d+2$. Some moment-preserving perturbation therefore moves any chosen order statistic.

Proposition~\ref{prop:impossibility} explains \emph{why} the operators excluded in Section~\ref{sec:closedforms} are excluded. The problem is not a missing derivation. Bounded-order moments simply do not carry enough information. Between the two regimes lies an open strip. Games that depend on moments, but not polynomially (e.g., standardized skewness, which divides a moment-polynomial numerator by a fractional power of the variance), can use the same summaries but have, to our knowledge, no exact Shapley closed form. Mapping this strip precisely remains future work.

\section{Experimental Evaluation}
\label{sec:eval}

We evaluate five questions. \textbf{RQ1:} Do the closed forms match exact Shapley values computed by exhaustive coalition enumeration? \textbf{RQ2:} How does per-slide attribution cost scale with window size $N$, compared to per-window recomputation, permutation sampling, and enumeration? \textbf{RQ3:} What accuracy does the sampling alternative buy at what cost, for ad-hoc predicates? \textbf{RQ4:} How do maintenance cost and state behave under many predicates and on real data at full stream rate? \textbf{RQ5:} Does the axiomatic weighting change the resulting explanations, relative to the naive lift heuristic and to the predicate-as-player meta-game semantics, on real data?

\subsection{Setup}
\label{sec:eval:setup}
All experiments run single-threaded on an Apple M2 Pro (16\,GB RAM) using a Python 3.12/NumPy 2.4 prototype. No compiled kernels are used, so absolute latencies are conservative. The meaningful quantity is the comparison across methods, which share data structures and predicate-evaluation code. The implementation and all experiment scripts are included with the submission. We compare four methods that answer the same query, namely the predicate attribution of a fixed predicate (selectivity $0.3$) for AVG and VAR\_POP at every slide of a count-based window:
\begin{itemize}
    \item \textbf{Incremental (ours):} incremental summary maintenance (Algorithm~\ref{alg:summary-maintenance}) plus closed-form evaluation; coefficients evaluated in $\mathcal{O}(1)$ from the running harmonic numbers $H_N$ and $H^{(2)}_N$ (constant for fixed count-window size $N$).
    \item \textbf{Per-window scan:} closed-form evaluation with summaries recomputed from scratch by a vectorized $\mathcal{O}(N)$ scan per window; the same $\mathcal{O}(1)$ harmonic-number coefficients. This isolates the value of incremental maintenance given our formulas.
    \item \textbf{MC ($M$):} stateless permutation sampling \cite{castro2009polynomial} with $M$ permutations per window; each permutation is walked once with incrementally maintained $(|S|,A,B)$, i.e., $\mathcal{O}(MN)$ per window, which is the strongest sampling variant for these games. Figure~\ref{fig:latency} reports the case $M=100$.
    \item \textbf{Enumeration ($2^{N}$):} exact enumeration of all $2^{N-1}$ predecessor coalitions per tuple via Eq.~\ref{eq:shapley_def}.
\end{itemize}
MC and Enumeration are generic methods that only see the game through its value function and do not exploit aggregate structure. We include them to show what the stateless alternatives cost, not as competitive baselines. The like-for-like comparison that isolates our contribution is Incremental versus Per-window scan: both use the same closed forms and differ only in how summaries are obtained. Explanation systems such as Scorpion~\cite{wu2013scorpion}, DIFF~\cite{abuzaid2021diff}, and TSExplain~\cite{chen2023tsexplain} answer related but non-axiomatic explanation queries. RQ5 therefore compares against the \emph{semantic} alternatives (the lift heuristic and the meta-game), which target the same quantity as we do. System-level comparisons on shared workloads are left to future work.

\subsection{RQ1: Exactness}
\label{sec:eval:exactness}
We compare closed-form outputs against Enumeration on a grid spanning $N \in \{2,\dots,12\}$, predicate selectivities $\{0,0.25,0.5,0.75,1\}$, and value distributions $\mathcal{N}(0,1)$ and $\mathrm{Uniform}[-100,100]$. This yields $10{,}010$ random windows and $30{,}030$ aggregate comparisons, checking every tuple-level value and the predicate-level attribution. We cap $N$ at 12 because enumeration costs $\mathcal{O}(N\cdot 2^{N-1})$. The closed forms are exact for all $N$ by Propositions~\ref{prop:avg} and~\ref{prop:variance}, so RQ1 validates the implementation. Table~\ref{tab:empirical_validation} reports the maxima. All deviations are at double-precision rounding level. The larger absolute deviations for the variance games under $\mathrm{Uniform}[-100,100]$ simply reflect the $10^{4}$-scale second moments involved. Relative deviations stay below $2\times10^{-12}$. The atomic-refinement Efficiency identity $\sum_{s\in\mathcal{S}_W}\Phi_{\alpha_s}(W)=\nu(W)-\nu(\emptyset)$ (Section~\ref{sec:overlap}), checked with $K=3$ random overlapping predicates on every window with $N\ge3$, holds to a maximum normalized deviation of $1.1\times10^{-12}$.

\begin{table}[t]
\small
\centering
\begin{tabular}{l l r r}
\toprule
\textbf{Aggregate} & \textbf{Distribution} & \textbf{Max Abs Dev} & \textbf{Max Rel Dev} \\
\midrule
AVG & $\mathcal{N}(0,1)$ & $1.9 \times 10^{-15}$ & $2.9 \times 10^{-15}$ \\
AVG & $\mathrm{Uniform}[-100,100]$ & $9.6 \times 10^{-14}$ & $4.2 \times 10^{-15}$ \\
VAR\_POP & $\mathcal{N}(0,1)$ & $3.6 \times 10^{-15}$ & $4.6 \times 10^{-13}$ \\
VAR\_POP & $\mathrm{Uniform}[-100,100]$ & $5.9 \times 10^{-12}$ & $1.8 \times 10^{-12}$ \\
VAR\_SAMP & $\mathcal{N}(0,1)$ & $4.2 \times 10^{-15}$ & $4.6 \times 10^{-13}$ \\
VAR\_SAMP & $\mathrm{Uniform}[-100,100]$ & $1.0 \times 10^{-11}$ & $1.5 \times 10^{-12}$ \\
\bottomrule
\end{tabular}
\caption{RQ1: maximum deviation between the closed forms and exhaustive coalition enumeration over $10{,}010$ random windows ($N\in[2,12]$; tuple- and predicate-level values both checked). All deviations are attributable to double-precision floating-point rounding in either method.}
\label{tab:empirical_validation}
\end{table}

\subsection{RQ2: Attribution Latency vs.\ Window Size}
\label{sec:eval:latency}
Figure~\ref{fig:latency} reports median per-slide latency as $N$ grows from $10^2$ to $10^6$ (stream of standard-normal values; 2{,}000 slides per configuration for Incremental and Per-window scan). Incremental is flat at $1.7$--$1.9\,\mu s$ per slide across four orders of magnitude of $N$, matching Corollary~\ref{cor:perslide}. Per-window scan grows linearly from $4.8\,\mu s$ to $5.9$\,ms. At $N=10^6$, Incremental is already $\approx 3{,}200\times$ faster than recomputing summaries per window \emph{with the same closed forms}. The stateless baselines are orders of magnitude further away. MC ($M=100$) costs $5.6$\,ms per window at $N=10^2$ and $576$\,ms at $N=10^4$ ($\approx 3\times10^{5}\times$ Incremental). Enumeration costs $11.6$\,ms at $N=10$, $0.53$\,s at $N=15$, and $24.1$\,s at $N=20$. Neither is viable at streaming rates even for small windows.

\begin{figure}
\centering
\includegraphics[width=0.7\linewidth]{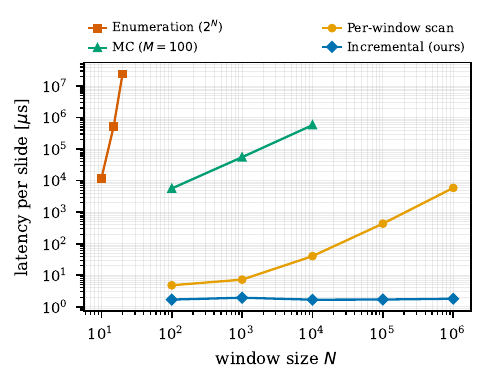}
\caption{RQ2: median per-slide latency for answering the same predicate-attribution query (AVG and VAR\_POP) as a function of window size $N$ (log--log). The closed-form incremental method is constant at $\approx 1.8\,\mu s$; per-window scanning grows linearly; sampling and enumeration are orders of magnitude slower and infeasible at streaming rates.}
\label{fig:latency}
\end{figure}

\subsection{RQ3: The Cost of Approximation}
\label{sec:eval:mc}
Because ad-hoc predicates unknown at maintenance time require either a scan or sampling (Section~\ref{sec:limitations}), we quantify the sampling route. Figure~\ref{fig:mc} shows the accuracy--latency tradeoff of MC as $M$ varies, at $N=1{,}000$ with values from $\mathcal{N}(10,1)$, so that AVG ($\approx 10$) and VAR\_POP ($\approx 1$) both have $\mathcal{O}(1)$ scale. Errors are means over 20 runs, relative to the aggregate value. The error decays as the expected $M^{-1/2}$. Reaching $\approx 0.7\%$ relative error requires $M=3{,}162$ permutations and $3.8$\,s per window, roughly $2\times10^{6}$ times the cost at which Incremental delivers the exact answer for maintained predicates. Sampling is thus a fallback for ad-hoc queries, not a competitor for registered ones.

\begin{figure}
\centering
\includegraphics[width=0.7\linewidth]{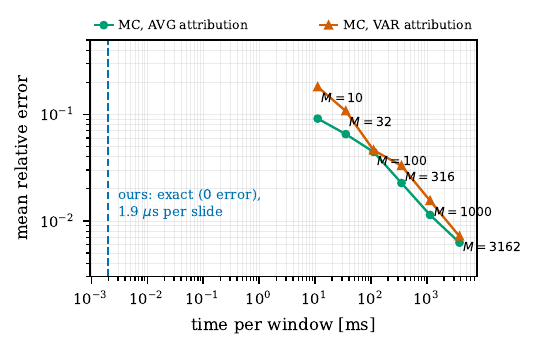}
\caption{RQ3: accuracy--latency tradeoff of permutation sampling for predicate attribution ($N=1{,}000$, values $\sim\mathcal{N}(10,1)$, mean of 20 runs). Errors decay as $M^{-1/2}$; the exact incremental closed form answers in $1.9\,\mu s$.}
\label{fig:mc}
\end{figure}

\subsection{RQ4a: Scaling in the Number of Predicates}
\label{sec:eval:scaling}
Table~\ref{tab:kscaling} scales the atomic-refinement state (Section~\ref{sec:overlap}) to $K=64$ registered predicates over categorical dimensions of cardinality 8 ($N=10^4$, 2{,}000 slides). Per-slide maintenance grows only from $0.88$ to $1.35\,\mu s$. Each slide touches exactly one atom per boundary tuple regardless of $K$ (the $\mathcal{O}(K)$ signature evaluation is a bit-packed vectorized step). The active-atom count $|\mathcal{S}_W|$ stays far below $2^K$ and saturates at its other bound $N$, as predicted in Section~\ref{sec:overlap}. The cost of decomposing the full window into all atom attributions scales with $|\mathcal{S}_W|$, not with $2^K$. The $K=64$ row also marks the honest boundary of the approach. With $|\mathcal{S}_W|\approx N$, the atom state amounts to one summary triple per active tuple, so the ``no tuple materialization'' advantage has been spent on predicate coverage. The framework targets moderate numbers of registered dimensions. Predicate search over unregistered candidates belongs to the ad-hoc regime of RQ3 (see also Section~\ref{sec:limitations}).

\begin{table}[t]
\small
\centering
\begin{tabular}{r r r r r}
\toprule
$K$ & $|\mathcal{S}_W|$ & $\min(2^K,N)$ & slide [$\mu s$] & all-atom query [$\mu s$] \\
\midrule
1 & 2 & 2 & 0.88 & 2.5 \\
4 & 5 & 16 & 0.88 & 2.2 \\
8 & 8 & 256 & 0.88 & 3.2 \\
16 & 64 & $10^4$ & 0.96 & 25.9 \\
32 & 3{,}740 & $10^4$ & 1.00 & 1{,}237 \\
64 & 9{,}994 & $10^4$ & 1.35 & 3{,}323 \\
\bottomrule
\end{tabular}
\caption{RQ4a: atomic-refinement scaling with the number of registered predicates $K$ ($N=10^4$). Maintenance stays near-constant; state and full-decomposition query cost track the \emph{active} signatures $|\mathcal{S}_W|$, not the $2^K$ worst case.}
\label{tab:kscaling}
\end{table}

\subsection{RQ4b: Case Study on NYC Taxi Data}
\label{sec:eval:casestudy}
We stream all 2{,}926{,}220 yellow-taxi trips of January 2024 from the NYC TLC trip records \cite{nyctlc2024} (fares filtered to $(0,500]$ dollars), through a time-based window of 1 hour emitting every 5 minutes, with $K=3$ registered predicates: \emph{airport} (pickup at JFK, LaGuardia, or Newark), \emph{manhattan} (pickup in Manhattan), and \emph{longtrip} (distance $>10$ miles, overlapping both).

\textbf{Throughput and state.} The full month (8{,}928 emit ticks; $N_t$ between 151 and 9{,}199, mean 3{,}931) is processed end-to-end in $3.4$\,s single-threaded. Summary updates take $3.3$\,s, i.e., $\approx 1.8$ million updates per second. Each emit tick then needs $4.8\,\mu s$ for the full atom decomposition under both AVG and VAR\_POP, of which $1.1\,\mu s$ is coefficient maintenance via the harmonic-number closed forms (Eqs.~\ref{eq:coeff_closed_1}--\ref{eq:coeff_closed_2}). There is no coefficient precomputation, no cache, and no $\mathcal{O}(N_t)$ step anywhere in the loop. Of the $2^3=8$ possible signatures, at most 6 atoms are ever active. The two signatures that would require an airport pickup inside Manhattan are impossible, exactly the state collapse anticipated in Section~\ref{sec:overlap}. Across all 8{,}928 real windows, the Efficiency identity holds to a relative deviation of at most $8.6\times10^{-11}$ despite a month of purely additive updates.

\textbf{Explaining a fare spike.} Figure~\ref{fig:casestudy} shows 24 hours starting 2024-01-09 18:00, with the active cardinality $N_t$ shown below so that scale effects (Section~\ref{sec:scale}) are visible next to the attributions. Between 02:00 and 06:00 the windowed average fare rises from a daytime level of $\approx\$15.6$ to a peak of $\$44.29$ (window at 04:25, $N_t=289$). The attribution decomposition explains the spike immediately and quantitatively. The airport group (128 trips) contributes $\Phi^{AVG}_{\mathrm{airport}}=+\$63.44$, \emph{more than the aggregate itself}, while the 136 Manhattan trips contribute $-\$35.73$. Nighttime Manhattan traffic is cheap and pulls the average down. The spike exists only because expensive airport trips dominate the thin overnight window. The share--lift decomposition (Eq.~\ref{eq:share_lift}) keeps this headline honest. With $\bar{y}_{\mathrm{airport}}=\$60.1$ and $w(289)\approx5.26$, the $+\$63.44$ splits into a $\$26.6$ share and a $\$36.8$ amplified lift. Note also that $N_t$ \emph{drops} overnight (from $\approx7{,}000$ to $\approx300$; bottom panel), which lowers $w(N_t)$. The spike in $\Phi$ is therefore driven by composition, not by the $\ln N$ weight. A compositional drill-down (Eq.~\ref{eq:compositional_query}) refines the explanation with no new state: $\Phi_{\mathrm{airport}\wedge\mathrm{longtrip}}=\$62.18$ versus $\Phi_{\mathrm{airport}\wedge\neg\mathrm{longtrip}}=\$1.26$. Long-haul airport runs, not short hops, carry the spike. The variance panel tells the complementary story during the spike. The non-airport, non-Manhattan group is the largest dispersion driver ($+2{,}182\,\$^2$). Manhattan trips are strongly \emph{stabilizing} ($-1{,}144\,\$^2$), illustrating the signed semantics of Section~\ref{sec:closedforms}. Airport trips contribute $+784\,\$^2$. All quantities are exact Shapley attributions computed from the maintained summaries at emit time, with no access to individual tuples.

\begin{figure}
\centering
\includegraphics[width=0.7\linewidth]{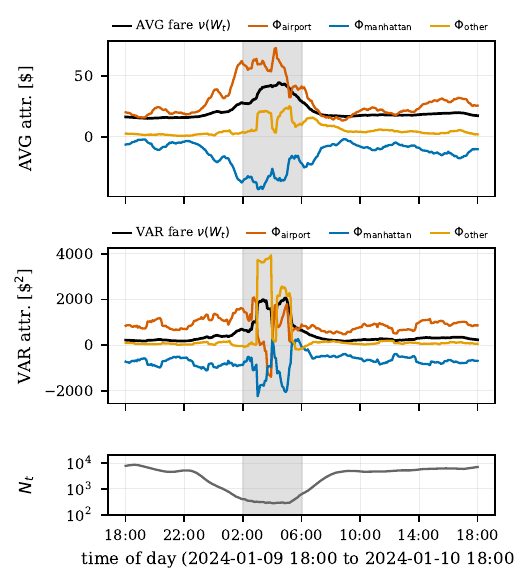}
\caption{RQ4b: 24 hours of NYC yellow-taxi fares (1\,h window, 5\,min emit). Top: windowed AVG fare and its exact Shapley decomposition into airport, Manhattan, and other pickups; the 02:00--06:00 spike (shaded) is driven by airport traffic, whose attribution exceeds the aggregate while Manhattan's is negative. Middle: the same decomposition for the VAR\_POP game; Manhattan trips stabilize the nighttime variance while dispersed outer-borough trips drive it. Bottom: active cardinality $N_t$ (log scale), overlaid because attributions carry a $\ln N_t$-weighted lift component (Section~\ref{sec:scale}).}
\label{fig:casestudy}
\end{figure}

\subsection{RQ5: Does the Axiomatic Weighting Matter?}
\label{sec:eval:lift}
Section~\ref{sec:scale} shows that AVG attribution combines share and lift with a $\ln N$ weight. Does the Shapley value ever \emph{disagree} with simply ranking groups by lift? We compare the two on all 8{,}928 real windows, for the three pickup groups (Figure~\ref{fig:lift}, top). They are highly correlated (Pearson $r=0.962$) and agree on the complete ranking in $90.8\%$ of windows, so the axioms mostly confirm the practitioner's heuristic, as one should hope. The disagreements, however, are not noise. In $522$ windows ($5.8\%$) the \emph{top-ranked} group differs, mostly in mid-size windows (median $N_t=878$ among disagreements). The cause is exactly the share term of Eq.~\ref{eq:share_lift}. For example, at 04:50 on 2024-01-02 ($N_t=239$, $\nu=\$34.73$), the lift heuristic names \emph{airport} the top driver, since it is the only group with positive lift ($\ell=+\$1.48$), while Shapley names \emph{manhattan} ($\Phi=\$21.52$ of the $\$34.73$ total). Here 186 of the 239 trips are Manhattan pickups, and they carry most of the aggregate's value despite below-average fares. The two answers address different questions. Lift asks ``who pushed the average above baseline?''; Shapley asks ``how does the current value distribute over groups?''. Only the latter satisfies Efficiency: over any partition the lifts sum to zero, so the heuristic cannot account for the value itself. Figure~\ref{fig:lift} (bottom) completes the scale story of Section~\ref{sec:scale}. At the fixed peak-window composition, $\Phi^{AVG}_{\mathrm{airport}}$ grows from $\$58.8$ to $\$88.1$ as $N$ ranges over the observed cardinalities, while the share component stays constant. This is why Figure~\ref{fig:casestudy} shows $N_t$ alongside the attribution curves.

\textbf{Sum-of-members vs.\ predicate-as-player.} We also computed, at every emit tick, the \emph{meta-game} semantics discussed in Section~\ref{sec:limitations}: the three groups act as three players and $\nu$ is evaluated on unions of their coalitions (exact via $3!$ permutations over merged summaries; $3.4\,\mu s$ per tick). The two semantics correlate ($r=0.831$) and agree on the top contributor in $93.2\%$ of windows, but their magnitudes differ substantially. The mean absolute deviation is $\$8.04$ per group ($39\%$ of the aggregate value on average) and the maximum is $\$62.26$. At the spike peak, the meta-game assigns airport $+\$25.97$ instead of $+\$63.44$ and Manhattan $-\$6.34$ instead of $-\$35.73$. With only three coarse players, each group's marginal contribution is evaluated against far larger coalitions, which compresses the attributions. Which semantics is preferable depends on the application. The sum-of-members semantics refines consistently under drill-down (Eq.~\ref{eq:compositional_query}) and is the one with constant-time maintenance. The meta-game depends on the chosen predicate granularity and costs $\mathcal{O}(2^K)$ per query for $K$ groups, exponential in general. Either way, the measured gap shows the choice is consequential, not cosmetic.

\begin{figure}
\centering
\includegraphics[width=0.7\linewidth]{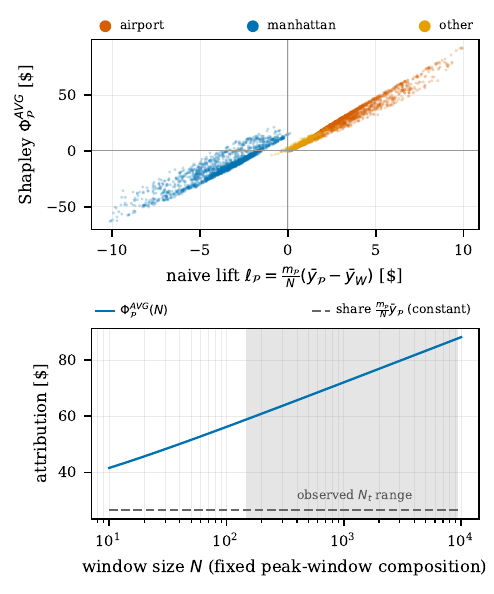}
\caption{RQ5. Top: exact Shapley attribution vs.\ naive lift for the three pickup groups over all 8{,}928 real windows (2{,}500 points sampled per group). The two agree broadly ($r=0.962$) but name a different top contributor in $5.8\%$ of windows. Bottom: $\Phi^{AVG}$ at fixed (peak-window) composition as a function of $N$: the share component is constant while the harmonically weighted lift grows like $\ln N$ across the observed cardinality range (shaded).}
\label{fig:lift}
\end{figure}

\section{Limitations and Future Work}
\label{sec:limitations}

The closed forms above give exact sufficient-statistics attribution for the aggregate games treated in this paper. Keeping the scope this narrow avoids unproved approximation assumptions. It also makes the limits of the results explicit.

\textbf{1. Aggregate scope.}
Our exact-maintenance results cover SUM, COUNT, AVG, both variance games, and, by Theorem~\ref{thm:sufficiency}, any moment-polynomial aggregate (including covariance and fixed central moments; Remark~\ref{rem:multivariate}). They do not cover quantiles, MIN/MAX (impossible by Proposition~\ref{prop:impossibility}), joins, or arbitrary user-defined functions. The territory between the two regimes is not fully mapped. Games that depend on moments non-polynomially (e.g., standardized skewness), ratio-of-moment statistics, DISTINCT-style aggregates, regression aggregates, and model-inference queries need new aggregate-specific arguments.

\textbf{2. Attribution semantics.}
Two semantic questions remain open. First, the closed forms attribute the \emph{current} aggregate value. Attributing the change $\nu(W_{t+1})-\nu(W_t)$ axiomatically would need a Shapley game over a common player set spanning both windows, which is an extra modeling choice. Second, an alternative to our sum-of-members semantics treats the predicates themselves as players in a meta-game over unions of their coalitions. Section~\ref{sec:eval:lift} measures the gap between the two on real data (high rank agreement, but magnitudes differing by $39\%$ of the aggregate on average). The meta-game is computable from the same atom summaries for small $K$ but exponential in general. When it is the better choice, and how to maintain it at scale, is left to future work.

\textbf{3. Registered and ad-hoc predicates.}
The constant-time guarantees assume that predicates (or dimensions) of interest are \emph{registered} in advance so that their summaries are maintained as the stream evolves. Predicates that first appear at query time need scans, indexes, or sampling instead. Lemma~\ref{lem:linearity} makes such approaches well-defined in principle, and Section~\ref{sec:eval:mc} shows how expensive the sampling route is. With very many registered predicates the atom state can also approach one summary per active tuple (Section~\ref{sec:eval:scaling}). A full systems treatment of ad-hoc predicates, indexing, and predicate promotion is left to future work.

\textbf{4. Temporal structure.}
The time-based analysis assumes that watermarks fix all insertions and expirations within each emit tick. Late arrivals need reordering or retroactive correction, which we do not formalize. More broadly, the window is treated as a bag of records. The temporal order of tuples plays no role beyond entry and exit, and extending the axioms to time-ordered Shapley values is a natural next step. Finally, Shapley attribution is not causal explanation. It divides the realized aggregate among tuples or atoms, but it does not identify a causal graph. If EU traffic drives premium-tier latency through a hidden dependency, both atoms can receive large variance attribution without the framework saying which caused which.

\section{Conclusion}
\label{sec:conclusion}
Explaining a streaming aggregate should not cost more than computing it. This paper showed that for SUM, COUNT, AVG, and both variance games, exact predicate-level Shapley attribution collapses to an affine function of three additively maintained summaries per predicate. Every attribution coefficient is a closed form in two running harmonic numbers. Arbitrary Boolean combinations of overlapping predicates are answered exactly from the atomic refinement, which also restores Efficiency. The phenomenon has a boundary on both sides: every moment-polynomial game admits such a closed form, while order statistics provably admit none of any fixed moment order. The share--lift decomposition makes the resulting numbers easy to interpret, including their logarithmic dependence on window size. Empirically, the attribution is exact to floating-point precision and costs a constant $\approx\!2\,\mu s$ per slide, independent of window size, in an unoptimized single-threaded prototype. This is $3{,}200\times$ faster than per-window recomputation with the same formulas at $N=10^6$, with stateless sampling and enumeration further orders of magnitude away. On real data it agrees with the naive lift heuristic on most windows but overturns its top-ranked explanation on $5.8\%$ of them. It also turns a real fare-spike anomaly on 2.9 million taxi trips into a quantitative, axiomatically grounded explanation at stream rate. Because the maintained state consists of ordinary associative aggregates, the method plugs into existing window-aggregation infrastructure. We believe it provides a practical foundation for always-on, explanation-aware stream processing.

\appendix
\section{Proof Details}

\subsection{AVG}
\label{app:avg}
In the permutation view of the Shapley value, each position of $x_i$ is equally likely. If $x_i$ appears first, the marginal contribution is $y_i$. If $x_i$ appears after $r\geq1$ predecessors, the predecessor set $S$ is a uniformly random size-$r$ subset of $W\setminus\{x_i\}$, and since $\mathbb{E}[A(S)]=r(A(W)-y_i)/(N-1)$,
\begin{equation}
    \mathbb{E}\left[\nu_{AVG}(S\cup\{x_i\})-\nu_{AVG}(S)\right]
    =
    \frac{1}{r+1}
    \left(y_i-\frac{A(W)-y_i}{N-1}\right).
\end{equation}
Averaging over $r=0,\dots,N-1$ and using $\sum_{r=1}^{N-1}1/(r+1)=H_N-1$ gives Eq.~\ref{eq:avg_tuple}. Summing over $x_i\in C_{\mathcal{P}}$ yields Eq.~\ref{eq:avg_predicate}. The derivation uses only the cardinality and the additive sum of the current window, so it transfers verbatim to time-based windows once $N$ is replaced by $N_t$. (This computation also appears, for static single-relation queries, in the proof of Proposition~5.2 of \cite{standke2025tractability}.)

\subsection{Variance}
\label{app:variance}

\paragraph{Moment expressions.}
For a fixed player $i$ and a uniformly random size-$r$ subset $S\subseteq W\setminus\{x_i\}$ drawn without replacement, each remaining tuple $x_j \in W\setminus\{x_i\}$ is included with probability $r/n$. Any pair $x_j,x_k$ with $j\neq k$ is jointly included with probability $r(r-1)/(n(n-1))$, where $n=N-1$. Therefore,
\begin{equation}
    \mathbb{E}[A(S)]=\frac{r}{n}T_i,
    \qquad
    \mathbb{E}[B(S)]=\frac{r}{n}U_i.
\end{equation}
For the second moment of $A(S)$,
\begin{equation}
    A(S)^2
    =
    \sum_{j \in S}y_j^2 + \sum_{j\neq k \in S} y_j y_k,
\end{equation}
so
\begin{equation}
\begin{aligned}
    \mathbb{E}[A(S)^2]
    &=
    \frac{r}{n}U_i
    +
    \frac{r(r-1)}{n(n-1)}\sum_{j\neq k}y_j y_k \\
    &=
    \frac{r}{n}U_i+
    \frac{r(r-1)}{n(n-1)}(T_i^2-U_i),
\end{aligned}
\end{equation}
where the last step uses $\sum_{j\neq k} y_j y_k = (\sum_j y_j)^2 - \sum_j y_j^2 = T_i^2 - U_i$ over $W\setminus\{x_i\}$.

\paragraph{$N=2$ case.}
When $N=2$, only the orderings $(x_i,x_j)$ and $(x_j,x_i)$ are possible, with $j$ the unique other tuple. Both predecessor sets give zero population variance (one is empty, the other is a singleton). The marginal contribution of $x_i$ averaged over the two positions is therefore $\tfrac{1}{2}\nu_{VAR\_POP}(\{x_i,x_j\})$. This matches Eq.~\ref{eq:var_tuple} evaluated at $N=2$ using the convention that the $r(r-1)/(n(n-1))$ term vanishes when $n-1=0$.

\paragraph{Population variance Shapley value.}
We expand the marginal contribution step-by-step:
\begin{equation}
\begin{aligned}
& \nu_{VAR\_POP}(S\cup\{x_i\}) - \nu_{VAR\_POP}(S) \\
&\quad= \frac{B(S)+y_i^2}{r+1}
  - \frac{A(S)^2 + 2y_i A(S) + y_i^2}{(r+1)^2} \\
&\qquad - \frac{B(S)}{r} + \frac{A(S)^2}{r^2} \\
&\quad= B(S)\Bigl(\tfrac{1}{r+1} - \tfrac{1}{r}\Bigr)
  + A(S)^2\Bigl(\tfrac{1}{r^2} - \tfrac{1}{(r+1)^2}\Bigr) \\
&\qquad - \frac{2y_i A(S)}{(r+1)^2}
  + y_i^2\Bigl(\tfrac{1}{r+1} - \tfrac{1}{(r+1)^2}\Bigr) \\
&\quad= -\frac{B(S)}{r(r+1)}
  + A(S)^2\Bigl(\tfrac{1}{r^2} - \tfrac{1}{(r+1)^2}\Bigr) \\
&\qquad - \frac{2y_i A(S)}{(r+1)^2}
  + \frac{r y_i^2}{(r+1)^2}.
\end{aligned}
\end{equation}
Taking the expectation over random predecessor coalitions $S$ of size $r$, we substitute $\mathbb{E}[B(S)] = \frac{r}{n}U_i$, $\mathbb{E}[A(S)] = \frac{r}{n}T_i$, and $\mathbb{E}[A(S)^2] = E_i^{(2)}(r)$, and simplify using $\frac{1}{r+1} - \frac{1}{(r+1)^2} = \frac{r}{(r+1)^2}$. Substituting these into the four terms above yields the expected marginal contribution $M_i(r)$:
\begin{equation}
\begin{aligned}
M_i(r)
&=
-\frac{U_i}{n(r+1)}
+
\Bigl(\frac{1}{r^2}-\frac{1}{(r+1)^2}\Bigr)E_i^{(2)}(r)\\
&\quad
-\frac{2ry_iT_i}{n(r+1)^2}
+\frac{r y_i^2}{(r+1)^2}.
\end{aligned}
\end{equation}
This matches Eq.~\ref{eq:var_delta}. Summing over $r$ and over $x_i\in C_{\mathcal{P}}$ uses
\begin{equation}
\begin{aligned}
\sum_{i\in C_{\mathcal{P}}} U_i &= U_{\mathcal{P}},\\
\sum_{i\in C_{\mathcal{P}}} y_i T_i &= YT_{\mathcal{P}},\\
\sum_{i\in C_{\mathcal{P}}} (T_i^2-U_i) &= Q_{\mathcal{P}},\\
\sum_{i\in C_{\mathcal{P}}} y_i^2 &= B_{\mathcal{P}}.
\end{aligned}
\end{equation}
where $U_{\mathcal{P}}$, $YT_{\mathcal{P}}$, and $Q_{\mathcal{P}}$ are defined in Eq.~\ref{eq:UP}--Eq.~\ref{eq:QP}. Substituting gives the predicate-level sum form
\begin{equation}
    \Phi_{\mathcal{P}}^{VAR}(W)
    =
    \frac{1}{N}
    \sum_{r=1}^{N-1}
    D_{\mathcal{P}}(r),
    \label{eq:var_predicate}
\end{equation}
with
\begin{equation}
\begin{aligned}
    D_{\mathcal{P}}(r)
    &=
    -\frac{U_{\mathcal{P}}}{n(r+1)}
    +
    \Bigl(\tfrac{1}{r^2}-\tfrac{1}{(r+1)^2}\Bigr)
    \Bigl[
        \tfrac{r}{n}U_{\mathcal{P}}
        +
        \tfrac{r(r-1)}{n(n-1)}Q_{\mathcal{P}}
    \Bigr] \\
    &\quad
    -
    \frac{2r}{n(r+1)^2}YT_{\mathcal{P}}
    +
    \frac{r}{(r+1)^2}B_{\mathcal{P}}.
\end{aligned}
    \label{eq:var_predicate_delta}
\end{equation}
Collecting the $r$-dependent factors yields Eq.~\ref{eq:var_predicate_constant} with
\begin{equation}
\begin{aligned}
    C_U &= \sum_{r=1}^{N-1}
    \Biggl[
        -\frac{1}{n(r+1)}
        +
        \Bigl(\tfrac{1}{r^2}-\tfrac{1}{(r+1)^2}\Bigr)\frac{r}{n}
    \Biggr],\\
    C_Q &= \sum_{r=1}^{N-1}
    \Bigl(\tfrac{1}{r^2}-\tfrac{1}{(r+1)^2}\Bigr)
    \frac{r(r-1)}{n(n-1)},\\
    C_{YT} &= \sum_{r=1}^{N-1}
    -\frac{2r}{n(r+1)^2},\\
    C_B &= \sum_{r=1}^{N-1}
    \frac{r}{(r+1)^2}.
\end{aligned}
\end{equation}
Partial fractions reduce each sum to harmonic numbers. For example, $\frac{r}{(r+1)^2}=\frac{1}{r+1}-\frac{1}{(r+1)^2}$, so $C_B$ telescopes to $(H_N-1)-(H^{(2)}_N-1)=G_N$. The same treatment of the remaining sums gives the closed forms Eq.~\ref{eq:coeff_closed_1}--Eq.~\ref{eq:coeff_closed_2}.

\paragraph{Sample variance.}
For sample variance and $r\geq2$,
\begin{equation}
\begin{aligned}
M_i^{SAMP}(r)
&=
\nu_{VAR\_SAMP}(S\cup\{x_i\})-\nu_{VAR\_SAMP}(S)\\
&=
\frac{B(S)+y_i^2-(A(S)+y_i)^2/(r+1)}{r}\\
&\quad-
\frac{B(S)-A(S)^2/r}{r-1}.
\end{aligned}
\end{equation}
Taking expectations and again grouping terms by their dependence on $y_i$ gives, for $r\ge 2$,
\begin{equation}
\begin{aligned}
    D_i^{SAMP}(r)
    &=
    -\frac{U_i}{n(r-1)}
    +
    \frac{2E_i^{(2)}(r)}{r(r-1)(r+1)}
    -
    \frac{2y_iT_i}{n(r+1)}
    +
    \frac{y_i^2}{r+1}.
\end{aligned}
    \label{eq:samp_rge2}
\end{equation}
The case $r=1$ is handled separately because $\nu_{VAR\_SAMP}(S)=0$ by definition for $|S|\le 1$. The expected marginal contribution of adding $x_i$ to a one-tuple predecessor set is
\begin{equation}
    D_i^{SAMP}(1)=
    \frac{1}{2}
    \left(
        \frac{U_i}{n}
        -
        \frac{2y_iT_i}{n}
        +
        y_i^2
    \right).
    \label{eq:samp_r1}
\end{equation}
The exact tuple-level Shapley value is
\begin{equation}
    \phi_i^{VAR\_SAMP}(W)=
    \frac{1}{N}
    \left[
        D_i^{SAMP}(1)
        +
        \sum_{r=2}^{N-1}D_i^{SAMP}(r)
    \right].
\end{equation}
Summing over tuple positions and over $x_i\in C_{\mathcal{P}}$ using the same predicate-summary identities as above gives Eq.~\ref{eq:varsamp_predicate_constant} with
\begin{equation}
\begin{aligned}
    C_U^{S}
    &=
    \frac{1}{2n}
    +
    \sum_{r=2}^{N-1}
    \Biggl[
        -\frac{1}{n(r-1)}
        +
        \frac{2}{n(r-1)(r+1)}
    \Biggr],\\
    C_Q^{S}
    &=
    \sum_{r=2}^{N-1}
    \frac{2}{(r+1)n(n-1)},\\
    C_{YT}^{S}
    &=
    -\frac{1}{n}
    +
    \sum_{r=2}^{N-1}
    -\frac{2}{n(r+1)},\\
    C_B^{S}
    &=
    \frac{1}{2}
    +
    \sum_{r=2}^{N-1}
    \frac{1}{r+1}.
\end{aligned}
\end{equation}
The identity $\frac{2}{(r-1)(r+1)}=\frac{1}{r-1}-\frac{1}{r+1}$ cancels the leading term of $C_U^S$ outright. The remaining sums telescope to the closed forms Eq.~\ref{eq:coeff_closed_s1}--Eq.~\ref{eq:coeff_closed_s2}.

\paragraph{$N=2$ sanity check.}
With $N=2$, only $D_i^{SAMP}(1)$ contributes. For $W = \{y_i, y_j\}$ with $n=1$, $U_i = y_j^2$ and $T_i = y_j$, giving $D_i^{SAMP}(1) = \frac{1}{2}(y_j^2 - 2y_iy_j + y_i^2) = (y_i-y_j)^2/2$, so $\phi_i = (y_i-y_j)^2/4$. By symmetry, $\phi_j = (y_i-y_j)^2/4$. The two values sum to $(y_i-y_j)^2/2 = \nu_{VAR\_SAMP}(\{y_i,y_j\})$, verifying Efficiency.

\subsection{Sufficiency and Impossibility}
\label{app:boundary}

\paragraph{Proof of Theorem~\ref{thm:sufficiency}.}
Fix $x_i \in W$ and let $V=W\setminus\{x_i\}$, $n=N-1$. In the permutation view, conditioned on $x_i$ arriving after $r$ predecessors, the predecessor set $S$ is a uniformly random $r$-subset of $V$, and
\begin{equation}
    \phi_i(\nu,W)=\frac{1}{N}\sum_{r=0}^{N-1}\mathbb{E}_{|S|=r}\left[\nu(S\cup\{x_i\})-\nu(S)\right].
\end{equation}
By Eq.~\ref{eq:mompoly} and $M_\ell(S\cup\{x_i\})=M_\ell(S)+y_i^{\ell}$, both $\nu(S)$ and $\nu(S\cup\{x_i\})$ are polynomials in the variables $\{M_\ell(S)\}_{\ell\le d}$ and $y_i$, of total weighted degree at most $D=dp$ (each factor $M_\ell$ or $y_i^{\ell}$ carrying weight $\ell$, each monomial in Eq.~\ref{eq:mompoly} having weight at most $dp$), with coefficients $g_\kappa(r)$ or $g_\kappa(r+1)$ depending only on $r$. It therefore suffices to show that for any exponents $q_1,\dots,q_d$,
\begin{equation}
    \mathbb{E}_{|S|=r}\Big[\textstyle\prod_{\ell=1}^{d}M_\ell(S)^{q_\ell}\Big]
\end{equation}
is a polynomial in the power sums $M_1(V),\dots,M_{D}(V)$ with coefficients depending only on $r$ and $n$. Expand the product into a sum over tuples of indices $(j_1,\dots,j_m)\in V^m$, $m=\sum_\ell q_\ell$, of monomials $y_{j_1}^{a_1}\cdots y_{j_m}^{a_m}$ with $\sum a_t \le D$, each multiplied by $\ind[j_1,\dots,j_m \in S]$. Group the tuples by the set partition induced by equality of indices. A tuple whose distinct indices are $k_1,\dots,k_q$ (with merged exponents $b_1,\dots,b_q$, $\sum b_u \le D$) contributes $y_{k_1}^{b_1}\cdots y_{k_q}^{b_q}\,\ind[k_1,\dots,k_q\in S]$. Under sampling without replacement,
\begin{equation}
\begin{aligned}
    \Pr[k_1,\dots,k_q \in S]
    &=\frac{(r)_q}{(n)_q}\\
    &=\frac{r(r-1)\cdots(r-q+1)}{n(n-1)\cdots(n-q+1)},
\end{aligned}
\end{equation}
a constant depending only on $(r,n,q)$. The remaining sum $\sum y_{k_1}^{b_1}\cdots y_{k_q}^{b_q}$ over \emph{distinct} $k_1,\dots,k_q \in V$ is, by inclusion--exclusion over the partition lattice (the standard expansion of augmented monomial symmetric functions in terms of power sums), a polynomial with integer coefficients in $M_{b}(V)$ for $b \le D$. Finally, $M_b(V)=M_b(W)-y_i^{b}$, so every term is a polynomial in $y_i$ (degree $\le D$) and the global power sums $M_1(W),\dots,M_D(W)$. Averaging over $r$ with weights $1/N$ and coefficients $g_\kappa(\cdot)$ preserves this structure, proving (i). Part (ii) follows by summing (i) over $x_i \in C_{\mathcal{P}}$, since the coefficients $c_\ell$ do not depend on $i$. Part (iii) is immediate because power sums are sums of per-tuple quantities. \hfill$\square$

\paragraph{Proof of Proposition~\ref{prop:impossibility}.}
Fix $d$ and let $N=d+2$. Consider the moment map $\mu:\mathbb{R}^{N}\to\mathbb{R}^{d}$, $\mu(z)=(M_1(z),\dots,M_d(z))$ with $M_\ell(z)=\sum_j z_j^{\ell}$. Its Jacobian has rows $\big(\ell z_1^{\ell-1},\dots,\ell z_N^{\ell-1}\big)_{\ell=1}^{d}$. At any point $z$ with pairwise distinct coordinates it has full rank $d$ (a generalized Vandermonde submatrix on any $d$ distinct coordinates is nonsingular), so $\ker J_\mu(z)$ has dimension $N-d=2$. Fix a target index $j^\ast$ (for MAX, the argmax; for a general order statistic, the coordinate holding that rank). The subspace $\{v \in \ker J_\mu(z): v_{j^\ast}=0\}$ has dimension at least $1$ less, i.e., at most $1 < 2$, so some $v \in \ker J_\mu(z)$ has $v_{j^\ast}\neq 0$. Since $\mu$ is a submersion at $z$, the fiber $\mu^{-1}(\mu(z))$ is locally a smooth $2$-manifold containing a curve $\gamma$ with $\gamma(0)=z$ and $\gamma'(0)=v$. For sufficiently small $t$, the coordinates of $\gamma(t)$ remain pairwise distinct and retain their ordering. The value of the chosen order statistic therefore equals $\gamma_{j^\ast}(t)$, which is non-constant in $t$ because $\gamma_{j^\ast}'(0)=v_{j^\ast}\neq0$. Taking $W=\{\gamma_j(0)\}_j$ and $W'=\{\gamma_j(t)\}_j$ for small $t\neq0$ yields two windows with identical $N$ and moments $M_1,\dots,M_d$ but different order-statistic values. If an attribution rule for the corresponding game were computed from $(N,M_1,\dots,M_d)$ and predicate moments alone, its total over any partition would agree on $W$ and $W'$ (take the trivial single-predicate partition, whose predicate moments coincide with the global ones). This contradicts Efficiency, which forces the totals $\nu(W)\neq\nu(W')$. Windows of any larger size $N>d+2$ follow by padding both windows with the same values. \hfill$\square$

\bibliographystyle{unsrtnat}
\bibliography{references}

\end{document}